\documentclass[11pt]{article}
\pdfoutput=1
\usepackage{amsmath, amssymb, amsfonts, amsthm}
\usepackage{graphicx,psfrag,epsf}
\usepackage{enumerate}
\usepackage{natbib}
\usepackage{url} 
\usepackage{fullpage}
\usepackage{xspace}
\usepackage{booktabs}
\usepackage{mathtools}
\usepackage{tabularx} 
\usepackage{hyperref}
\usepackage{multirow}
\usepackage{algorithm}
\usepackage{algorithmic}
\usepackage{subcaption}
\usepackage{color}
\usepackage{booktabs}
\usepackage{bm, stmaryrd}
\usepackage{comment}
\usepackage{verbatim}

\newtheorem{definition}{Definition}
\newtheorem{theorem}{Theorem}

\newtheorem{assumption}{Assumption}
\newtheorem{lemma}{Lemma}

\newtheorem{proposition}{Proposition}
\newtheorem{remark}{Remark}

\newcommand{\sdp}[1]{\ensuremath{X_{#1}}}
\newcommand{\bkin}{\ensuremath{\beta_k^{(in)}}\xspace}
\newcommand{\bkout}{\ensuremath{\beta_k^{(out)}}\xspace}
\newcommand{\km}{k-means\xspace}

\newcommand{\bfone}{\textbf{1}}

\newcommand{\cF}{\mathcal{F}}

\newcommand{\cM}{\mathcal{M}}

\newcommand{\cS}{\mathcal{S}}
\newcommand{\Z}{Z}

\newcommand{\bE}{\mathbb{E}}
\newcommand{\bR}{\mathbb{R}}

\newcommand{\bC}{\mathbb{C}}

\newcommand{\lamn}{\lambda_n}
\newcommand{\lamz}{\lambda_0}
\newcommand{\eig}{\theta}

\newcommand{\er}{Erd\H{o}s-R\'{e}nyi}
\renewcommand{\dim}{d}

\newcommand{\tr}{\text{trace}}

\newcommand{\diag}{\text{diag}}
\newcommand{\xhat}{\hat{X}}
\newcommand{\xh}{\hat{X}}

\newcommand{\rd}{\color{red}}
\newcommand{\bk}{\color{black}}

     \DeclareMathOperator*{\argmax}{arg\,max}
     
\newcommand{\pmax}{p_{\max}}
\newcommand{\pmin}{p_{\min}}    
\newcommand{\mmin}{m_{\min}} 
\newcommand{\mmax}{m_{\max}}
\newcommand{\dmin}{d_{\min}}

\newcommand{\cl}{r}
\newcommand{\rom}[1]{%
  \textup{\uppercase\expandafter{\romannumeral#1}}%
}

\newcommand{\opnorm}[1]{\|#1\|_{\text{op}}}
\newcommand{\llnorm}[1]{\|#1\|_{\ell_{\infty}\to\ell_1}}
\newcommand{\innerprod}[2]{\langle {#1}, {#2}\rangle}

\newcommand{\vertiii}[1]{{\left\vert\kern-0.25ex\left\vert\kern-0.25ex\left\vert #1 
    \right\vert\kern-0.25ex\right\vert\kern-0.25ex\right\vert}}
    
\newcommand{\bas}[1]{\begin{align*}#1\end{align*}}
\newcommand{\ba}[1]{\begin{align}#1\end{align}}
\newcommand{\beq}[1]{\begin{equation}#1\end{equation}}
\newcommand{\bsplt}[1]{\begin{split}#1\end{split}}

\definecolor{alizarin}{rgb}{0.82, 0.1, 0.26}

\begin{document}
\title{\bf  Covariate Regularized Community Detection \\
in Sparse Graphs}
  \author{Bowei Yan and 
    Purnamrita Sarkar \\
    University of Texas at Austin}
    \date{}
  \maketitle

\begin{abstract}
 In this paper, we investigate community detection in networks in the presence of node covariates.  In many instances, covariates and networks individually only give a partial view of the cluster structure. One needs to jointly infer the full cluster structure by considering both. In statistics, an emerging body of work has been focused on  combining information from both the edges in the network and the node covariates to infer community memberships. However, so far the theoretical guarantees have been established in the dense regime, where the network can lead to perfect clustering under a broad parameter regime,  and hence the role of covariates is often not clear.   In this paper, we examine sparse networks in conjunction with finite dimensional sub-gaussian mixtures as covariates under moderate separation conditions.  In this setting each individual source can only cluster a non-vanishing fraction of nodes correctly.  We propose a simple optimization framework which provably improves clustering accuracy when the two sources carry partial information about the cluster memberships, and hence perform poorly on their own. Our optimization problem can be solved using scalable convex optimization algorithms. Using a variety of simulated and real data examples, we show that the proposed method outperforms other existing methodology.
\end{abstract}

\noindent%

\section{Introduction}
\label{sec:intro}
Community detection in networks is a fundamental problem in machine learning and statistics. A variety of important practical problems like analyzing socio-political ties among leading politicians~\citep{gil1996political}, understanding brain graphs arising from diffusion MRI data~\citep{binkiewicz2014covariate}, investigating ecological relationships between different tiers of the food chain~\citep{jacob2011role}  can be framed as community detection problems. Much attention has been focused on developing models and methodology to recover latent community memberships. Among generative models, the stochastic block model \citep{holland1983stochastic} and its variants (\cite{airoldi2009mixed} etc.) have attracted a lot of attention, since their simplicity facilitates efficient algorithms and asymptotic analysis~\citep{rohe2011spectral,amini2013pseudo,chen2014statistical}.

Although most real world network datasets come with covariate information associated with nodes, existing approaches are primarily focused on using the network for inferring the hidden community memberships or labels. 
 Take for example the Mexican political elites network (described in detail in Section \ref{sec:exp}). This dataset comprises of  35 politicians  (military or civilian)  and their connections. The associated covariate for each politician is the year when one came into power. After the military coup in 1913, the political arena was dominated by the military. In 1946, the first civilian president since the coup was elected. Hence those who came into power later are more likely to be civilians. Politicians who have similar number of connections to the military and civilian groups are hard to classify from the network alone. Here the temporal covariate is crucial in resolving which group they belong to. On the other hand, politicians who came into power around 1940's, are ambiguous to classify using covariates. Hence the number of connections to the two groups in the network helps in classifying these nodes. Our method can successfully classify these politicians and has higher classification accuracy than existing methods~\citep{binkiewicz2014covariate,  zhang2015community}.

In Statistics literature, there has been some interesting work on combining covariates and dense networks (average degree growing faster than logarithm of the number of nodes). 
In~\cite{binkiewicz2014covariate}, the authors present assortative covariate-assisted spectral clustering (ACASC) where one does Spectral Clustering on the the gram matrix of the covariates plus the regularized graph Laplacian weighted by a tuning parameter. A joint criterion for community detection (JCDC) with covariates is proposed by~\cite{zhang2015community}, which could be seen as a covariate reweighted Newman-Girvan modularity. This approach enables learning different influence on each covariate. In concurrent work~\citet{weng2016community} provide a variational approach for community detection.

All of the above works are carried out in the dense regime with strong separability conditions on the linkage probabilities. ACASC also requires the number of dimensions of covariates to grow with the number of nodes for establishing consistency.

In contrast, we prove our result for sparse graphs where the average degree is constant and the the covariates are finite dimensional sub-gaussian mixtures with moderate separability conditions. In our setting, neither source can yield consistent clustering in the limit. We show that combining the two sources leads to improved upper bounds on clustering accuracy under weaker conditions on separability on each individual source.

Leveraging information from multiple sources have been long studied in Machine learning and Data mining under the general envelop of multi-view clustering methods.  
~\citet{kumar2011co} use a regularization framework so that the clustering adheres to the dissimilarity of clustering from each view.
~\citet{liu2013multi} optimize the nonnegative matrix factorization loss function on each view, plus a regularization forcing the factors from each view to be close to each other.  The only provable method is by~\cite{chaudhuri2009multi}, where the authors obtain guarantees where the two views are mixtures of Log-concave distributions. This algorithm does not apply to networks. 

In this paper, we propose a penalized optimization framework for community detection when node covariates are present. We take the sparse degree regime of Stochastic Blockmodels, where one can only correctly cluster a non-vanishing fraction of nodes. Similarly, for covariates, we assume that the covariates are generated from a finite dimensional sub-gaussian mixture with moderate separability conditions. 
We prove that our method leads to an improved clustering accuracy under weaker conditions on the separation between clusters from each source. As byproducts of our theoretical analysis we obtain new asymptotic results for sparse networks under weak separability conditions and kernel clustering of finite dimensional mixture of sub-gaussians.  Using a variety of real world and simulated data examples, we show that our method often outperforms existing methods. Using simulations, we also illustrate that when the two sources only have partial and in some sense orthogonal information about the clusterings, combining them leads to better clustering than using the individual sources. 

In Section~\ref{sec:setup}, we introduce relevant notation and present our optimization framework.  In Section~\ref{sec:kernel},  we present our main results, followed by experimental results on simulations and real world networks in Section~\ref{sec:exp}. Majority of the proofs are presented in the appendix, with details deferred to the supplementary.

\section{Problem Setup}
\label{sec:setup}

In this section, we introduce our model and set up the convex relaxation framework. For clarity, we list all definitions and notations that will be used later in Table~\ref{tab:notations}.

Assume $(C_1,\cdots, C_{\cl})$ represent a $\cl$-partition for $n$ nodes $\{1,\cdots,n\}$. Let $m_i=|C_i|$ be the size of cluster $i$, and let  $m_{\min}$ and $m_{\max}$ be the minimum and maximum cluster sizes respectively. We use $\pi_i:=\frac{m_i}{n}$, $\pi_{\min}=\frac{m_{\min}}{n}$ and $\alpha=\mmax/\mmin$.
We denote by $A$ the $n\times n$ binary adjacency matrix and by $Y$ the $n\times \dim$ matrix of $\dim$ dimensional covariates. The generation of $A$ and $Y$ share the true and unknown membership matrix $\Z=\{0,1\}^{n\times \cl}$. We define the graph model as: 
\begin{align}
\label{eq:data-model}
\mbox{(Graph Model)}\qquad P(A_{ij}=1|Z) = Z_i^TBZ_j  \quad\mbox{For $i\ne j$}
\end{align}
 $B$ is a $\cl \times \cl$ matrix of within and across cluster connection probabilities. Furthermore $A_{ii}=0,\forall i\in[n]$.  We consider the sparse regime where $n\max_{k\ell} B_{k\ell}$ is a constant and hence average expected degree is also a constant w.r.t $n$.
~\citet{amini2014semidefinite} define two different classes of block models in terms of separability properties of $B$. We state this below. 
\begin{definition}
	A stochastic block model is called \textit{strongly assortative} if $\min_k B_{kk}>\max_{k\ne \ell} B_{k\ell}$.  
	It is called \textit{weakly assortative} if $\forall k\ne \ell,\ B_{kk}>B_{k\ell}.$ 
\end{definition}
This distinction is important because the weakly assortative class of blockmodels is a superset of strongly assortative models, and most of the analysis are done in the stronger setting. 
To our knowledge, there has not been any work on weakly assortative blockmodels in the sparse setting. 
For the covariates, we define,
\begin{equation}
\mbox{(Covariate Model)}\qquad Y_i = \sum_{a=1}^r Z_{ia} \mu_{a}+  W_i
\label{eq:data-model-high}
\end{equation}
$W_i$ are mean zero $\dim$ dimensional sub-gaussian vectors with spherical covariance matrices $\sigma_k^2 I_\dim$  and sub-gaussian norm $\psi_k$ (for $i\in C_k$). 
Standard definitions of sub-gaussian random variables (for more detail see~\citet{vershynin2010introduction}) are provided in the Supplementary material.
We define
the distance between clusters $C_k$ and $C_\ell$ as
$ d_{k\ell}=\|\mu_k-\mu_\ell\|$ and the separation as $\dmin=\min_{k\ne \ell} d_{k\ell}.$

\begin{table}[H]
\centering
\footnotesize
\begin{tabular}{|l|l|l|}
\hline
Notation & Mathematical Definition & Explanation\\
\hline
$n,d$ 	& 	& Number of nodes, dimensionality of covariates\\
$I_d$	&	& identity matrix of size $d\times d$\\
$\mbox{diag}(v_1,\dots, v_k)\in \bR^{k\times k}$&& Diagonal matrix with diagonal $(v_1,\dots, v_k)$\\
$\cl$ & $\Theta(1)$ & Number of clusters \\
$B\in [0,1]^{\cl\times\cl}$ &$\Theta(1/n)$&Symmetric Probability matrix in SBM\\
$Z\in \{0,1\}^{n\times \cl}$ & & Latent class memberships\\
$m_i$&$\sum_j Z(j,i)$&Number of points in $i$th cluster\\
$\pi_i$&$\frac{m_i}{n}$& Proportion of points in $i$th cluster\\
$m_{\max}$&$\max_k m_k$, $\Theta(n)$ & Largest cluster size\\
$m_{\min}$&$\min_k m_k$, $\Theta(n)$ & Smallest cluster size\\
$\alpha$&$m_{\max}/\mmin$, $\Theta(1)$ & Ratio between largest and smallest clusters\\
$C_k$ &$\{j:Z(j,i)=1\}$& Point set for $k$th cluster\\
$X_0\in \bR^{n\times n}$ & $Z\mbox{diag}(1/m_1,\dots,1/m_\cl)Z^T$& Ground truth clustering matrix\\
$a_k,b_k=\Theta(1)$		&$a_k=nB_{kk}, b_k=n\max_{\ell\ne k}B_{k\ell}$	& Rescaled probabilities\\
$g\in \bR$& $\frac{2}{n-1}\sum_{i<j} Var(A_{ij})$, $\Theta(1)$\bk & Average variance of Graph edges\\
$\mu_k, \sigma_kI_d$& & Mean, covariance matrix for $Y_i$ if $i\in C_k$\\
$\psi_k$& & subgaussian norm for $Y_i$ if $i\in C_k$\\
$d_{k\ell}$ & $\|\mu_k-\mu_{\ell}\|$& Distance between cluster centers for the covariates\\
$K_I$ &Eq.~\eqref{eq:K_I}& Reference matrix for the kernel \\
$\nu_k$ &Eq.~\eqref{eq:kernel-lodim-sep}& Separation in $K_I$\\
$\gamma$& $\min_k (a_k-b_k+\lambda \nu_k)$, $\Theta(1)$ & Separation  of $ZBZ^T+\lambda K$\\
\hline
\end{tabular}
\caption{Population quantities used in the paper}
\label{tab:popnotations}
\end{table}
\paragraph{Notation}
For a matrix $M\in \bR^{n\times n}$, we use $\|M\|_F$ and $\|M\|$ to denote the Frobenius and operator norms of $M$ respectively. The $\ell_\infty$ norm is defined as:
$\|M\|_\infty=\max_{i,j} |M_{ij}|$. 
For two matrices $M, Q \in \bC^{m\times n}$, their inner product is $\innerprod{M}{Q}=\tr(M^TQ)$. 
The $\ell_\infty\to \ell_1$ norm of a matrix $M$ is defined as
$\|M\|_{\ell_\infty\to \ell_1} = \max_{\|s\|_{\infty}\le 1} \|Ms\|_1$. 
From now on we use $I_{n}$ to denote the identity matrix of size $n$, $\bfone_n$ to represent the all one $n$-vector and $E_n,E_{n,k}$ to represent the all one matrix with size $n\times n$ and $n \times k$ respectively. We use standard order notations $O,o,\Omega,\omega$, etc. For example, we use $t(n)=\Theta(1/n)$ to denote that $t(n)\times n$ is a constant w.r.t $n$. We also use $\tilde{O}$ notation to exclude multiplicative terms that are logarithmic in $n$. 
\begin{table}
	\centering
	\footnotesize
	\begin{tabular}{|l|l|l|}
		\hline
		Notation & Mathematical Definition & Explanation\\
		\hline
		$A\in \{0,1\}^{n\times n}$ & $A_{ij}|i\in C_k,j\in C_\ell \sim Ber(B_{k\ell})$ & Adjacency matrix (Symmetric)\\
		$Y_i\in \bR^d$&&Covariate observation for $i$th point\\
		$K\in [0,1]^{n\times n}$ & $K(i,j)=f(\|Y_i-Y_j\|_2^2)$ & Kernel matrix, symmetric and positive definite\\
		\hline
	\end{tabular}
	\caption{Random variables used in the paper}
	\label{tab:notations}
\end{table}

\begin{table}[t]
	\centering
	\footnotesize
	\begin{tabular}{|l|l|l|}
		\hline
		Notation & Mathematical Definition & Explanation\\
		\hline
		$\bfone_n$				& 					&All one vector of length $n$\\
		$E_n$					& $\bfone_n\bfone_n^T$	&All ones matrix of size $n\times n$\\
		$I_d$					&					& Identity matrix of size $d\times d$\\
		$K_G$ 					& $\le 1.783$ 			& Grothendieck's constant\\
		$f(x):\bR_+\rightarrow [0,1] $	&$\exp(-\eta x)$ 		& Kernel function\\
		\multirow{2}{*}{$\mathcal{F}$ }	& $\{X\succeq 0,\ \ 0\leq X\leq \frac{1}{m_{\min}},$ 	& \multirow{2}{*}{Feasible set of the SDP}\\
								&$ X\bfone_n=\bfone_n, \ \ \tr(X)=\cl\}$			&\\
		$\sdp{M}$ 				& $\arg\max_{X}\innerprod{M}{X} \  \ s.t.\ \ X\in \mathcal{F}$ 		& Solution matrix of the SDP\\
		$\rd \eig_i{(M)}\bk$ 				& 		& $i$-th eigenvalue of $M$\\
		$\lamn, \lamz$		&$\lamn=\lamz/n$, $\lamz = \Theta(1)$			&Tuning parameter between graph and covariates\\
		\hline
	\end{tabular}
	\caption{Useful notations and definitions}
	\label{tab:notations}
\end{table}
\vspace{-1em}
\subsection{Optimization Framework}
We now present our optimization framework.
There are many available semidefinite programming (SDP) relaxations for clustering blockmodels~\citep{amini2014semidefinite, cai2015robust, chen2014statistical}. The common element in all of these is maximizing the inner product between $A$ and $X$,
for a positive semidefinite matrix $X$. 
Here $X$ is a stand-in for the clustering matrix $ZZ^T$.  
Unequal-sized clusters is usually tackled with an extra regularization term added to the objective function (see \cite{hajek2016achieving, perry2015semidefinite, cai2015robust} among others). 
  While the above consistency results are for dense graphs,  \citet{guedon2014community,Montanari:2016} show that in the sparse regime one can use this method to obtain an error rate which is a constant w.r.t $n$ and depends on the gap between the within and across cluster probabilities.

SDPs are not only limited to network clustering.  Several convex relaxations for k-means type loss are proposed in the literature (see \citet{peng2007approximating,mixon2016clustering,yan2016robustness} for more references). In particular in these settings one maximizes $\innerprod{W}{X}$, for some positive semidefinite matrix $X$, where $W$ is a matrix of similarities between pairwise data points. For classical $k$-means $W_{ij}$ can be $Y_i^TY_j$ whereas for $k$-means in the kernel space one uses a suitably defined kernel similarity function between the $i$th and $j$th covariates. We analyze the widely-used Gaussian kernel  to allow for non-linear boundaries between clusters. Let $K$ be the $n\times n$ kernel matrix whose $(i,j)$-th entry is $K(i,j)=f(\|Y_i-Y_j\|_2^2)$, where $f(\cdot)$, where $f(x)=\exp(-\eta x)$ for $x\geq 0$.
This kernel function is upper bounded by 1 and is Lipschitz continuous w.r.t. the distance between two observations. Furthermore, in contrast to network based SDPs, the above uses $X$ as a stand in for the normalized variant of the clustering matrix $ZZ^T$, i.e.  the desired solution is $(X_0)_{ij} = \frac{1(k=\ell)}{m_k}$, if $i\in C_k,j\in C_\ell$.
It can be seen that $\|X_0\|_F^2=\cl$.

In our optimization framework, we propose to add a $k$-means type regularization term to the network objective, which enforces that the estimated clusters are consistent with the latent memberships in the covariate space. 
 \ba{
 	X =\arg\max_{X}\innerprod{A+\lamn K}{X} \  \ s.t.\ \ X\in \mathcal{F},
 	\label{eq:sdpcomb}
 }
 where
$\lamn$ is a tuning parameter (possibly depending on $n$) and the constraint set $\mathcal{F}=\{X\succeq 0,\ \ 0\leq X\leq \frac{1}{m_{\min}},\ \ X\bfone_n=\bfone_n, \ \ \tr(X)=\cl\}$ is similar  to~\cite{peng2007approximating}. The $m_{\min}$ in the constraint can be replaced by any lower bound on the smallest cluster size, and is mainly of convenience for the analysis. In the implementation, it suffices to enforce the elementwise positivity constraints, and other linear constraints. For ease of exposition, we define \ba{
	\sdp{M} =\arg\max_{X}\innerprod{M}{X} \  \ s.t.\ \ X\in \mathcal{F},
	\label{eq:sdpgen}
}

When $K(i,j)=Y_i^TY_j$, then the non-convex variant of the objective function naturally assumes a form similar to the work of ACASC (modulo normalization of $A$). 


\section{Main Results}
\label{sec:kernel}

Typically in existing SDP literature for sparse networks or subgaussian mixtures~\citep{guedon2014community,mixon2016clustering}, one obtains a relative error bound of the deviation of $X_M$ (the solution of the SDP ) from the ideal clustering matrix $X_0$. This relative error is typically proportional to the ratio of the observed matrix with a suitably defined reference matrix, and some quantity which measures the separation between the different clusters. 
Our theoretical result shows that the relative error of the solution to the combined SDP is  proportional to the ratio of the observed $A+\lamn K$ matrix to a suitably defined reference matrix to a quantity which measures separation between clusters. This quantity is a non-linear combination of the separations stemming from the two sources.
We first present an informal version of the main result. 
	{\it
		Main theorem (informal): 
		Let $X_{A+\lamn K}$ be the solution of SDP~\eqref{eq:sdpgen}. Let $s_G^{k}$ and $s_C^{k}$ be constants denoting the separations of cluster $k$ from the other clusters defined in terms of the model parameters of the network and the covariates respectively. If the tuning parameter $\lamn = \lamz/n$ for some constant $\lamz$, then
		$$\|X_{A+\lamn K}-X_0\|_F^2 \le \frac{c_G+\lamz c_C}{\min_k \left(s_G^{k}+\ell s_C^{k}\right)},$$ where $c_G$ and $c_C$ are constants representing the error corresponding to the graph and the covariates.
	}

Note that in SBM, the separation is well-defined, i.e. when $M=A$, a natural choice of the reference matrix is $E[A|Z]$ which is blockwise constant. In this case, the separation is given by $\min_k (B_{kk}-\max_{\ell}B_{k\ell})$, and leads to a result on weakly assortative sparse block models which we present in more details in Section~\ref{sec:graph}. However, for the kernel matrix $K$, the main difficulty is that one cannot achieve element-wise or operator norm concentration of $K$ (also discussed in~\cite{von2008consistency}). This makes the choice of the reference matrix difficult. 
To better understand the role of the separation parameter, we first present a key technical lemma bounding $\|\sdp{M}-X_0\|_F$. The main goal of this lemma is to establish an upper bound on the frobenius norm difference between the solution to an SDP with input matrix $M$ to the ideal clustering matrix. 
\bk
\begin{lemma}
	Let \sdp{M} be defined by Eq~\eqref{eq:sdpgen} for some input matrix $M$. Also let $Q$ be a reference matrix where $Q_{ij}=\bkin,\forall i,j\in C_k$, and $\bkout\ge Q_{ij} \ge 0,\forall i\in C_k,j\in C_\ell,k\ne \ell$. If $\min_k(\bkin-\bkout)\ge 0$, then 
	\ba{\label{eq:genlem}
		\| \sdp{M} -X_0\|_F^2 \le 2\frac{\innerprod{M-Q}{ \sdp{M} -X_0}}{\mmin \min_k(\bkin-\bkout)}  
	}
	\label{lem:fro_to_innerprod}
\end{lemma}

\begin{remark}
The key to the above lemma is to find a suitable reference matrix $Q$ which satisfies some separation conditions between the blocks. The deviation between $X_M$ and $X_0$ is small if $M-Q$ is small, and large if the separation between blocks in $Q$ is small.	While the proof technique is inspired  by~\citet{guedon2014community}, the details are different because of our use of different constraints and because our reference matrix $Q$ does not have to be blockwise constant and can be weakly assortative instead of strongly assortative.
	\end{remark}
The results on networks, covariates and the combination of the two essentially reduces to identifying good reference matrices ($Q$) for the input matrices $A$, $K$, and $A+\lambda K$, which 
\begin{enumerate}
	\item Satisfies the properties of $Q$ in the above lemma.
	\item Has a large separation $\min_k(\bkin-\bkout)$  increasing the denominator of Eq.~\eqref{eq:genlem}.
	\item Has a small deviation from $M$, thereby reducing the numerator of Eq~\eqref{eq:genlem}.
\end{enumerate}  

%
Now the main work is to choose the reference matrix $Q$ for $A+\lambda K$. As pointed out before, a common choice for reference matrix of $A$ is $\bE[A|Z]$.
For the covariates, we divide the nodes into ``good'' nodes $\cS_k:=\{i\in C_k:\|Y_i-\mu_k\|\leq \Delta_k\}$ and the rest. Also define $\cS=\cup_{k=1}^\cl \cS_k$.  $\Delta_k$ will be defined such that the kernel matrix induced by the rows and columns in $\cS$ is weakly assortative, and $3\Delta_k+\Delta_\ell\leq d_{k\ell}$. Define
\begin{align}
\label{eq:kernel-lodim-sep}
r_k:= f(2\Delta_k),\quad s_k:=\max_{\ell\neq k }f(d_{k\ell}-\Delta_k-\Delta_\ell),\qquad \nu_k =r_k-s_k
\end{align}
A simple use of triangle inequality gives $\min_{i,j\in \cS_k} K_{ij}\ge r_k$ and $\max_{i\in \cS_k, j\in S_\ell, \ell\neq k} K_{ij} \le s_k$. 
Hence the separation for cluster $k$ is $\nu_k:=r_k-s_k$.
We define the reference matrix $K_I$ as:
\beq{
	(K_I)_{ij}=\left\{ \begin{matrix} f(2\Delta_k), & \mbox{ if $i,j\in C_k$}\\ \min\{f(d_{k\ell}-\Delta_k-\Delta_\ell),K_{ij}\}, & \mbox{ if $i\in C_k, j\in C_\ell, k\ne \ell$}\end{matrix}\right. 
	\label{eq:K_I}
}
The choice of $\Delta_k$ is crucial. A large $\Delta_k$ makes the size of non-separable nodes $\cS^c$ small, but drives down the separation $ \nu_k$. 

We are now ready to present our main result.
As we will show in the proof, the new separation is $\gamma=\min_k \frac{(a_k-b_k)+\lamz \nu_k}{n}$. Typically, in the general case with unequal sub-gaussian norms, one should benefit from using different $\Delta_k$'s for different clusters. For example for a cluster with a large $a_k-b_k$, we can afford to have a small $\nu_k$. To think in terms of $\Delta_k$, for this cluster one can have a large $\Delta_k$, which will make $|\cS_k|$ larger than before, but will not affect the separation $(a_k-b_k)+\lamz \nu_k$ of cluster $k$  very detrimentally.  We now present our first main theorem. 

\begin{theorem}
	Let $a_k=nB_{kk}, b_k=n\max_{\ell\ne k}B_{k\ell}$, $g:=\frac{2}{n-1}\sum_{i<j}\text{Var}(a_{ij})\ge 9$. Take $\lamn=\lamz/n$, $m_k=n\pi_k$, $\mmin=n\pi_{\min}$, and $\pi_0:=\sum_k (m_k\exp(-\Delta_k^2/5\psi_k^2)+ \sqrt{m_k\log m_k/2})/n$. Let $\sdp{A+\lamn K}$ be defined as in Eq~\eqref{eq:sdpgen}. If $\pi_{\min}=\Theta(1)$ and  $\min_k(a_k-b_k+\lamz \nu_k)>0$, then, with probability tending to one,
	\bas{
		\| \sdp{A+\lambda K} -X_0\|_F^2 \le 2K_G\frac{ 6\sqrt{g}+\lamz \left( 2 \pi_0+\sum_{k}\pi_k^2(1-f(2\Delta_k)) \right) }{\pi_{\min}^2 \min_k(a_k-b_k+\lamz \nu_k)},
	}
	where
	$\nu_k=f(2\Delta_k)-\max_{\ell\ne k}f(d_{k\ell}-\Delta_k-\Delta_\ell)$ for some $\Delta_k,\Delta_\ell\geq 0$ and $\max(\Delta_k,\Delta_\ell)\leq d_{k\ell}/4$.
	\label{th:sparse_low}
\end{theorem}
Here $K_G$ is the Grothendieck's constant. The best value of $K_G$ is still unknown, and the best known bound is $K_G \le 1.783$~\citep{braverman2013grothendieck}.
First note that in the sparse case, we take $\lamn=\lamz/n$ for some constant $\lamz$.   In general the upper bound depends on several parameters such as $\lamn$ and the scale parameter $\eta$ in the gaussian kernel. We provide procedures for tuning $\lamn$ and $\eta$ in Section~\ref{sec:exp}. The $\Delta_k$'s show up in the numerator as well as the denominator.  Finding the optimal $\Delta_k$ is cumbersome in the general case with unequal $\psi_k$'s. In Section~\ref{sec:cov} we derive an upper bound for equal $\Delta_k$'s for concreteness.

Now we present two natural byproducts of our analysis, namely the result on graphs, i.e.  bounds on $\|X_0-\sdp{A}\|_F$ and the result on covariate clustering i.e. bounds on $\|X_0-\sdp{K}\|_F$.

\subsection{Result on Sparse Graph}
\label{sec:graph}
While most dense network-based community detection schemes give perfect clustering in the limit~\citep{amini2013pseudo,amini2014semidefinite,cai2015robust,chen2014statistical,yan2017exact}, in the sparse case no algorithm is consistent; however semidefinite relaxations (among others) can achieve an error rate governed by the within and across cluster probabilities~\citep{guedon2014community,Montanari:2016}. The sparse network analysis is done under strongly assortative settings.  

\begin{proposition}[Analysis for graph]
	Let $a_k, b_k$ defined as in Theorem~\ref{th:sparse_low} are positive constants and $g\geq 9$. Then with probability tending to 1, 
	\bas{
		\frac{\|\sdp{A}-X_0\|_F}{\|X_0\|_F} \le \epsilon,
	}
	if  $\min_k (a_k-b_k)\geq \frac{ 23 \alpha^2\cl \sqrt{g}}{\epsilon^2}$
	where $\alpha:=m_{\max}/m_{\min}$.
	\label{prop:main_graph_sparse}
\end{proposition}
Note that in the above result, in order to have the error rate $\epsilon$ to go to zero, one would require $a_k-b_k$ to go to infinity, whereas by definition $a_k,b_k$ are constants. Therefore one can only hope for a small albeit constant $\epsilon$. In addition, both number of clusters $\cl$ and the ratio between largest and smallest cluster sizes $\alpha$ needs to be constant order w.r.t $n$ in order to guarantee the error rate does not increase when the network grows.
\begin{remark}[Comparison with prior work]
	In contrast to having $\min_k a_k-\max_k b_k$ (strong assortativity) in the denominator like~\citet{guedon2014community}, we have $\min_k (a_k-b_k)$  (weak assortativity), which allows for a much broader parameter regime. 
\end{remark}

\subsection{Result on Covariates}
\label{sec:cov}
We present a result for covariates analogous to the sparse graph setting, which establishes that, while SDP with covariates is not consistent with finite signal-to-noise ratio, it achieves a small error rate if the cluster centers are further apart. But before delving into our analysis, we provide a brief overview of existing work.

For covariate clustering, it is common to make distributional assumptions; usually a mixture model with well-separated centers suffices to show consistency. 
The most well-studied model is Gaussian mixture models, which can be inferred by Expectation-Maximization algorithm, for which recently there has been some local convergence results~\citep{balakrishnan2014statistical,yan2017convergence} and its variants~\citep{dasgupta2007probabilistic}. The condition required for provable recovery on the separation is usually 
 the minimum distance between clusters is greater than some multiple of the square root of dimension (or effective dimension). 

Another popular technique is based on SDP relaxations. For example, \citet{peng2007approximating,mixon2016clustering} propose a SDP relaxation for \km type clustering. 
To make the analysis concrete, for Proposition~\ref{prop:lowdim_kernel}, we use $\Delta_k=\Delta$.

\begin{proposition}[Analysis for Covariates]
Let $K$ be the kernel matrix generated from kernel function $f$. Denote $\nu_k$ as in Eq~\eqref{eq:kernel-lodim-sep}. 
If $\frac{\dmin}{\psi_{\max}}>\max\left\{\sqrt{d},\frac{180}{\sqrt{d}} \right\}$, then 
	with properly chosen $\eta$, with probability at least $1-\sum_k \frac{1}{m_k}$,
	\bas{
		\frac{\| \sdp{K} -X_0\|_F^2}{\|X_0\|_F^2} \le& C\alpha^2 d\frac{\psi_{\max}^2}{\dmin^2} \max\left\{\log\left(\frac{\dmin}{\psi_{\max}\sqrt{d}}\right),r\right\}
	}
	\label{prop:lowdim_kernel}
\end{proposition}
\begin{remark}[Comparison with prior work]
 In recent work, \citet{mixon2016clustering} show the effectiveness of SDP relaxation with \km clustering for sub-gaussian mixtures, provided the minimum distance between centers is greater than the standard deviation of the sub-gaussian times the  number of clusters $\cl$. We provide a dimensionality reduction scheme, which also shows that the separation condition requires that  $\dmin=\Omega( \sqrt{\min(\cl,\dim)})$. Our proof technique is new and involves carefully constructing a reference matrix for Lemma~\ref{lem:fro_to_innerprod}.
\end{remark}

\subsection{Analysis of Covariate Clustering when $\dim \gg \cl$}
\label{sec:dim_red}

In high dimensional statistical problems, the signal is often assumed to lie in a low dimensional subspace or manifold. This is why much of Gaussian Mixture modeling literature first computes some projection of the data onto a low dimensional subspace~\citep{vempala2004spectral}. To reduce the dimensionality of the raw data, one could do a feature selection for the covariates (e.g.~\citet{jin2017phase, verzelen2017detection}).  
In contrast, here we propose a much simpler dimensionality reduction step, which does not distort the pairwise distances between cluster means too much. 
 The intuition is that, for clustering a subgaussian mixture,  if $\dim \gg \cl$, the effective dimensionality of the data is $\cl$ since the cluster means lie in an at most  $\cl$-dimensional subspace.

Hence we  propose the following simple dimensionality reduction algorithm when $\dim\gg\cl$ in a spirit similar to~\citet{chaudhuri2009multi}. 
We split up the sample into two random subsets $P_1$ and $P_2$ of sizes $n_1$ and $n-n_1$ and compute the top $\cl-1$ eigenvectors $U_{\cl-1}$ of the matrix $\hat{S}=\frac{\sum_{i\in P_1} (Y_i-\bar{Y}) (Y_i-\bar{Y})^T}{n_1}\in \mathbb{R}^{\dim\times \dim}$, where $\bar{Y}=\frac{\sum_{i\in P_1} Y_i}{n_1}$. Now we project the covariates from subset $P_2$ onto this lower dimensional subspace as $Y_i'=U_{\cl-1}^T Y_i$ to get the low dimensional projections. We take $n_1=n/\log n$.

\begin{lemma}
	\label{lem:dimreduce}
	Let $M:=\sum_k \pi_k \mu_k\mu_k^T$. If  $\sum_k\pi_k\mu_k=0$,  and the smallest eigenvalue of $M$ satisfies $\eig_{\cl-1}(M) \geq 5\psi_{\max}^2+C\sqrt{\frac{d\log^2n}{n}}$ for some constant $C$,  the projected $Y'_i$ are  also independent data points generated from an isotropic sub-gaussian mixture in $\cl-1$ dimensions. Furthermore the minimum distance between the means in the $\cl -1$ dimensional space is at least $d_{\min}/2$ with probability at least $1-\tilde{O}(\cl^2n^{-d})$, where  $d_{\min}$ is the separation in the original space.
\end{lemma}

The proof of this lemma is deferred to the supplementary material. We believe the proof can be generalized to non-spherical cases as long as the largest eigenvalue of covariance matrix for each cluster is bounded.
Typically $\eig_{\cl-1}(M)$ signifies the amount of signal. For example, for the simple case of mixture of two gaussians with $\pi_1=1/2$, and $\mu_2=-\mu_1$, $\eig_{\cl-1}(M)=\|\mu_1\|^2$, which is essentially $d_{\min}^2/4$. Hence the condition on $\eig_{\cl-1}(M)$ essentially translates to a lower bound on the signal to noise ratio, i.e. $d_{\min}^2\geq 48\psi_{\max}^2+C'\sqrt{\frac{d\log^2 n}{n}}$ for some constant $C'$. When $d>\cl$, if one applies Lemma~\ref{lem:dimreduce} on the $\cl -1$ dimensional space, then as long as $d_{\min}^2 = \Omega(\psi_{\max}^2\cl)$, the separation in the low dimensional space also satisfies the separation condition in Proposition~\ref{prop:lowdim_kernel}. Thus the dimensionality reduction  brings down the separation condition in Proposition~\ref{prop:lowdim_kernel} from $\Omega(\psi_{\max}\sqrt{\dim})$ to $\Omega(\psi_{\max}\sqrt{\min(\cl,\dim)})$. 
\bk
	
	The sample splitting is merely for theoretical convenience which ensures that the projection matrix and the projected data are independent, resulting in the fact that the final projection is also an independent sample from a sub-gaussian mixture.  To be concrete, the labels of $P_1$ do not matter asymptotically, since they incur a relative error in $\|X_0-\sdp{K}\|_F/\|X_0\|_F$ less than $\sqrt{n^2/(m_{\min}^2\log n)}/\sqrt{\cl}\leq \sqrt{\alpha^2\cl/\log n}$, where $\alpha$ and $r$ are both constants.
	In our setting, the relative error in Proposition~\ref{prop:lowdim_kernel} is a small but non-vanishing constant, and so this additional vanishing error term does not affect it.  However this sample splitting step  is not necessary in practice~\citep{chaudhuri2009multi}, and so we do not pursue this further. 
	\bk

We now present the tuning procedure, and experimental results.

\section{Experiments}
\label{sec:exp}

\begin{figure}[t!]
\centering
\begin{tabular}{ccc}
\hspace{-15pt}\includegraphics[width=.3\textwidth]{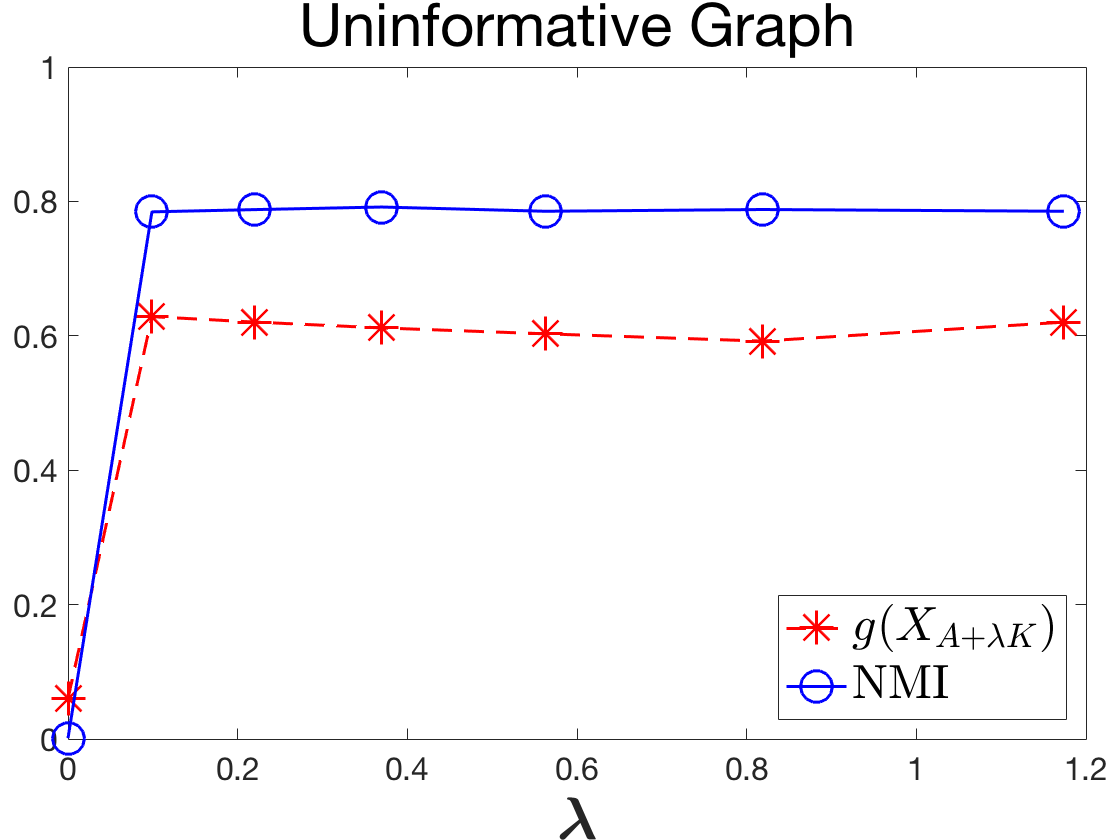} &
\hspace{-15pt}\includegraphics[width=.3\textwidth]{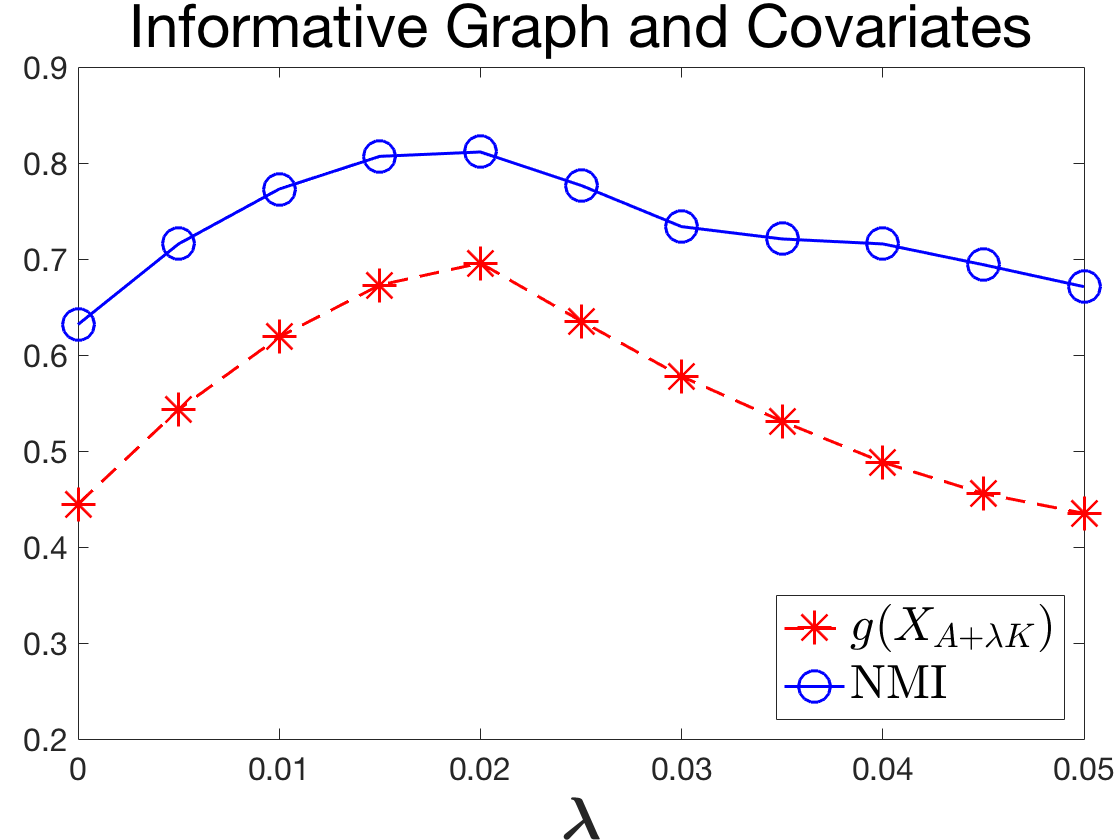} &
\hspace{-15pt}\includegraphics[width=.3\textwidth]{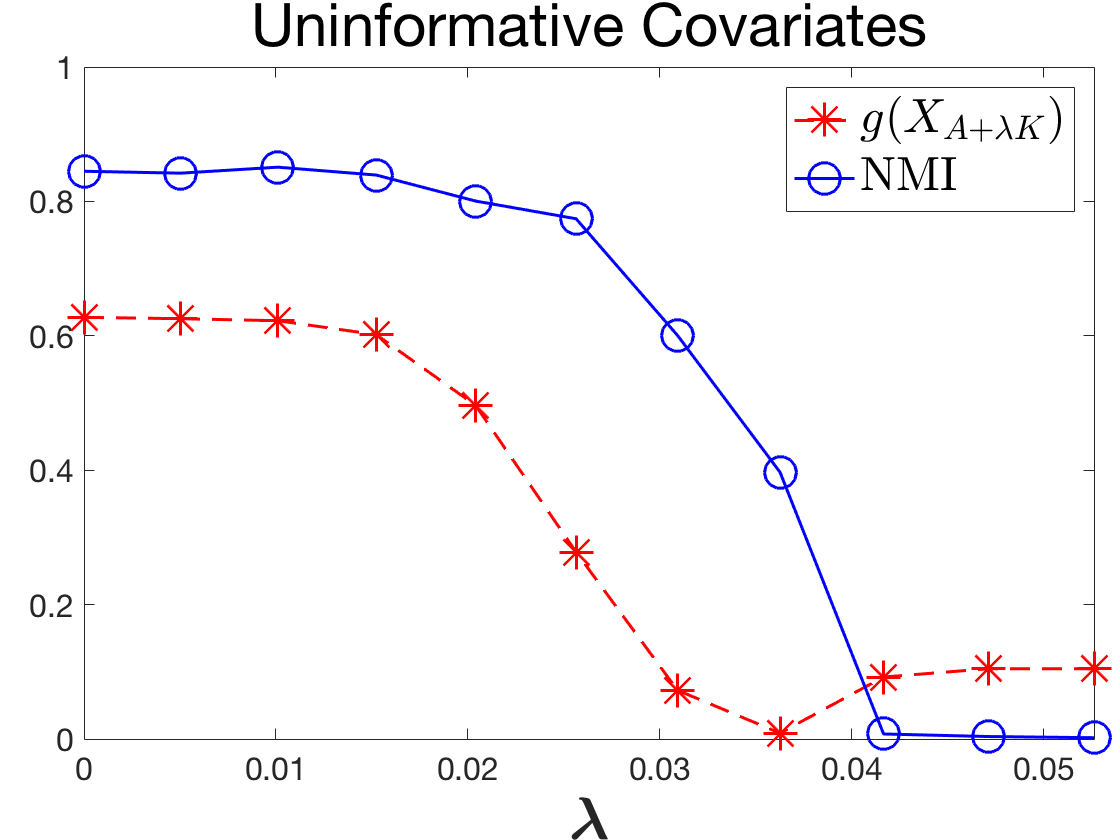}\\
(a)  & (b) & (c) 
\end{tabular}
\caption{Tuning: 
(a) $B = 0.005E_3, n=1000$, $\dim=6,\dmin=15\sigma$;
(b) $\dim=6, \dmin=1.3, \sigma=(1,1,5), B=\diag(0.004,0.024,0.024)+0.004E_3$;
(c) $\dim=6, \dmin=0, B = 0.0144I_3+0.0016E_3$.
}
\label{fig:tuning}
\end{figure}
In this section, we present results on real and simulated data. 
The cluster labels in our method are obtained by spectral clustering of the solution matrix returned by the SDP.
We will use SDP-comb, SDP-net, SDP-cov to represent the labels estimated from \sdp{A+\lamn K}, \sdp{A} and \sdp{K} respectively. 
Performance of the clustering is measured by normalized mutual information (NMI), which is defined as the mutual information of the two distributions divided by square root of the product of their entropies. 
We have also calculated classification accuracy and they show similar trends, so only NMI is reported in this section. 
For real and simulated data, we compare: 
(1) Covariate-assisted spectral clustering (ACASC) \citep{binkiewicz2014covariate}; (2) JCDC \citep{zhang2015community}, (3) SDP-comb, (4) SDP-net and (5) SDP-cov. The last two are used as references of graph-only and covariate-only clustering respectively.

\subsection{Implementation and computational cost}
Solving semidefinite programming with linear and non-linear constraints has been a challenging problems in numerical optimization community. Many SDPs proposed in statistical literature~\citep{cai2015robust, chen2014statistical, amini2014semidefinite} are solved by the alternating descent method of multipliers (ADMM) algorithm~\citep{boyd2011distributed}. Although ADMM is tractable for middle-sized problems and reasonable numerical behavior, whether it convergences in presence of non-negative constraints, which is prevalent in network literatures, remains an open problem. Recently, \citet{yang2015sdpnal} propose a majorized semismooth Newton-CG augmented Lagrangian method, called SDPNAL+, which is provably convergent. We solve the SDP using the matlab package of SDPNAL+ in all our experiments\footnote{The code used for the experiment can be found at \url{https://github.com/boweiYan/SDP_SBM_unbalanced_size}.}. The package provides an efficient implementation of the algorithm. Solving the SDP for matrix of size $1000\times 1000$ takes less than a minute on a Macbook with a 1.1 GHz Intel Core M processor.

\subsection{Choice of Tuning Parameters}
\label{sec:tuning}

As we pointed out earlier, the elementwise upper bound $\frac{1}{\mmin}$ is only for convenience of theoretical analysis. In the implementation, we do not enforce this constraint. So the main tuning parameters would be the scale parameter in the kernel matrix $\eta$ and the tradeoff parameter between graph and covariates $\lamn$. In most of our experiments the number of clusters is assumed known. In this section, we also provide a practical way to choose among candidates of $\cl$ when it is not given.
\paragraph{Choice of $\eta$}

We use the method proposed in \citet{shi2009data} to select the scale parameter. The intuition is to keep enough (say $10\%$) of the data points in the ``range" of the kernel for most (say $95\%$) data points.
Given the covariates, we first compute the pairwise distance matrix. Then for each data point $Y_i$, compute $q_i$ as $10\%$ quantile of $d(Y_i, Y_j), \forall j\in [n]$.
The bandwidth is defined as 
\bas{
w = \frac{95\%\mbox{ quantile of $q_i$}}{\sqrt{95\%\mbox{ quantile of }\chi^2_d}}
}
and scale parameter $\eta = \frac{1}{2w^2}$.

Note when the data is high-dimensional, we will first conduct dimensionality reduction as in Section~\ref{sec:dim_red}, then use the intrinsic dimension to tune the scale parameter.

\paragraph{Choice of $\lamn$}
As $\lamn$ increases, the resulting \sdp{A+\lamn K} clustering gradually changes from \sdp{A} clustering to \sdp{K} clustering. Our theoretical results show that, with the right $\lamn$, $ \sdp{A+\lamn K}$ and $X_0$ should be close, and hence also have similar eigenvalues. 
Let $\eig_i(M)$ be the $i$-th eigenvalue of matrix $M$.
Define the eigen gap function for clustering matrices $g(X):=(\eig_\cl(X)-\eig_{\cl+1}(X))/\eig_\cl(X)$. 
Using Weyl's inequality and the fact that  $\opnorm{ \sdp{A+\lamn K} -X_0} \le \| \sdp{A+\lamn K} -X_0\|_F $, we have: $\eig_\cl(X_0)-\| \sdp{A+\lamn K} -X_0\|_F \le \eig_\cl( \sdp{A+\lamn K})\le\eig_\cl(X_0)+\|  \sdp{A+\lamn K} -X_0\|_F$.
Since $g(X_0)=1$, we pick the $\lamn$ maximizing $g( \sdp{A+\lamn K})$. In Figure \ref{fig:tuning} (a)-(c), figures from left to right represent the situation where graph is uninformative (\er), both are informative and covariates are uninformative. We plot  $g( \sdp{A+\lamn K})$ and NMI of the clustering from $ \sdp{A+\lamn K}$ with the true labels against $\lamn$. 
Figure \ref{fig:tuning} shows that $g( \sdp{A+\lamn K})$ and NMI of the predicted clustering have a similar trend, justifying the effectiveness of the tuning procedure.


\paragraph{Unknown number of clusters}
In many real world settings, it is generally hard to possess the knowledge of number of clusters. Methods are proposed for selecting number of blocks under sparse stochastic block models~\citep{le2015estimating}, but most of these methods are designed specific for graph adjacency matrix and cannot be generalized to continuous matrix scenarios. We observe that the eigen gap acts as an informative indicator for picking the number of clusters. So when the number of clusters is unknown, we run the SDP over a grid of ${\lamn, k}$, and choose the pair that maximizes the eigen gap.
As we show in Figure~\ref{fig:unknownr}, we construct two settings and test the performance of using eigen gap to select $r$. In the first setting, the true model has 3 clusterings with proportion $3:4:5$, the probability matrix is {\footnotesize $B=0.01*\begin{bmatrix}1.6 &1.2 &0.16\\
        1.2& 1.6& 0.02\\ 
        0.16& 0.02& 1.2\end{bmatrix}$}. And the covariates are high dimensional Gaussian centered at $\mu_1=(0,2,0\cdots, 0)$, $\mu_2=( -1,-0.8,0\cdots, 0)$, $\mu_3=( 1,-0.8,0\cdots, 0)$. 
        We sample $n=800$ data points, and run SDP on top of it with different choice of $\lamn$ and specified number of clusters $k$. For each pair of parameter, we compute the NMI and eigengap and plot them on the upper and lower panel of Figure~\ref{fig:unknownr}(a). As we can see, the eigen gap presents a similar trend as the NMI, hence picking the pair that optimizes eigen gap will have a relatively high NMI as well. Note here the mis-specified $k=2$ has a higher NMI than that of the true value of $r$. This tells us even the number of clusters is mis-specified, the SDP is still able to find structure that correlates with the underlying model. This phenomenon is also observed in several other works \citep{yan2017exact, perry2015semidefinite}.
    
In the second scenario, we generate a planted partition model with 10 equal-sized clusters, where $B=0.046I_{10}+0.004E_{10}$, along with Gaussian covariates centered at $[3*I_{10} \ | \ \bm{0}_{3,90}]$. We conduct the same type of experiment as above and plot the NMI and eigengap. In this case, the eigen gap succussfully recovered the true number of clusters. 
    
\begin{figure}
\centering
\begin{tabular}{cc}
\includegraphics[width=.45\textwidth, height=6cm]{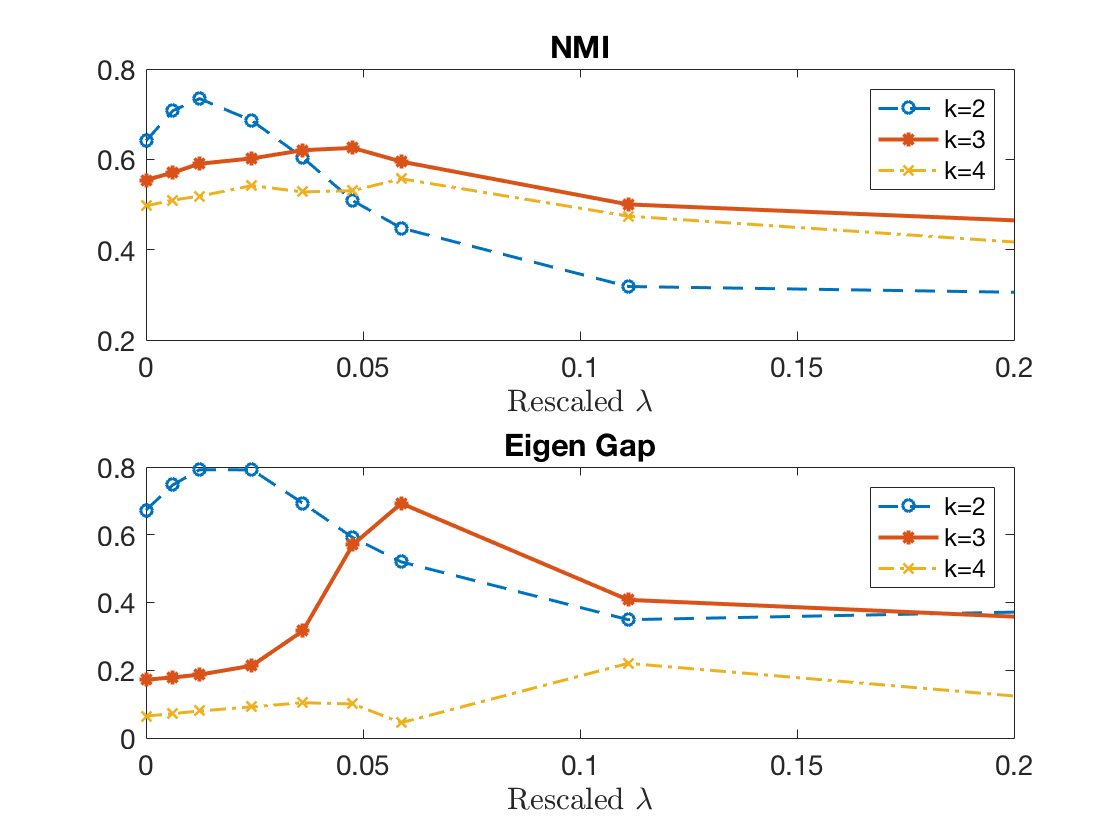}&
\includegraphics[width=.45\textwidth, height=6cm]{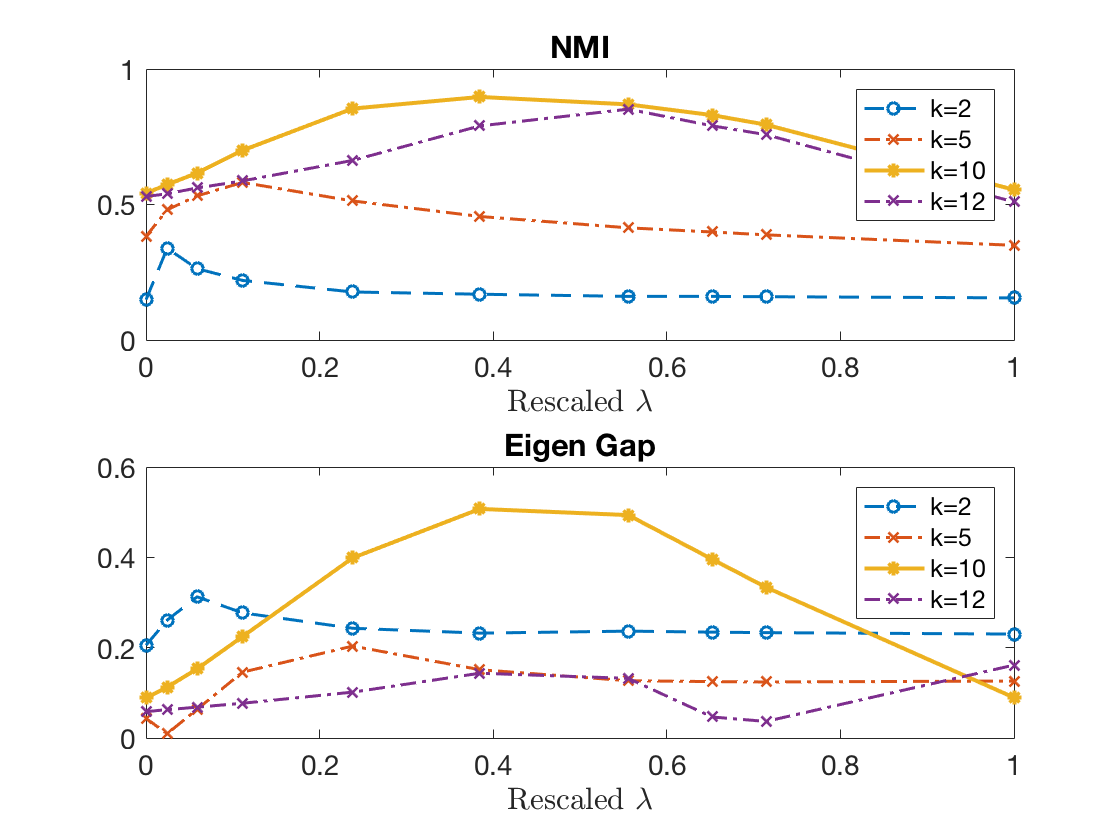}\\
(a) true $r=3$ & (b) true $r=10$.
\end{tabular}
\caption{NMI and eigen gap for various choice of $\cl$.}
\label{fig:unknownr}
\end{figure}

\subsection{Simulation Studies}
\label{sec:simulation}
In this part we consider two simulation settings. In the first setting, we generate three clusters with sizes 3:4:5, with $n=800$. The probability matrix is {\footnotesize $B=0.01*\begin{bmatrix}1.6 &1.2 &0.16\\
        1.2& 1.6& 0.02\\ 
        0.16& 0.02& 1.2\end{bmatrix}$}, and the covariates for each cluster are generated with $100$ dimensional unit variance isotropic Gaussians, whose centers are only non-zero on the first two dimensions with $\mu_1=(0,2,0\cdots, 0)$, $\mu_2=( -1,-0.8,0\cdots, 0)$, $\mu_3=( 1,-0.8,0\cdots, 0)$. This is the same setting as in the first simulation for unknown $r$.
         In this example, the network cannot separate out clusters one and two well, whereas the covariates can. On the other hand, clusters two and three are not well separated in the covariate space, while they are well separated using the network parameters.
 The experiments are repeated on 10 independently generated samples and the box plot for NMI is shown as in Figure~\ref{fig:simu1}(c).
\begin{figure}[!ht]
\centering
\begin{tabular}{@{\hspace{-2em}}ccc}
\includegraphics[width=0.35\textwidth]{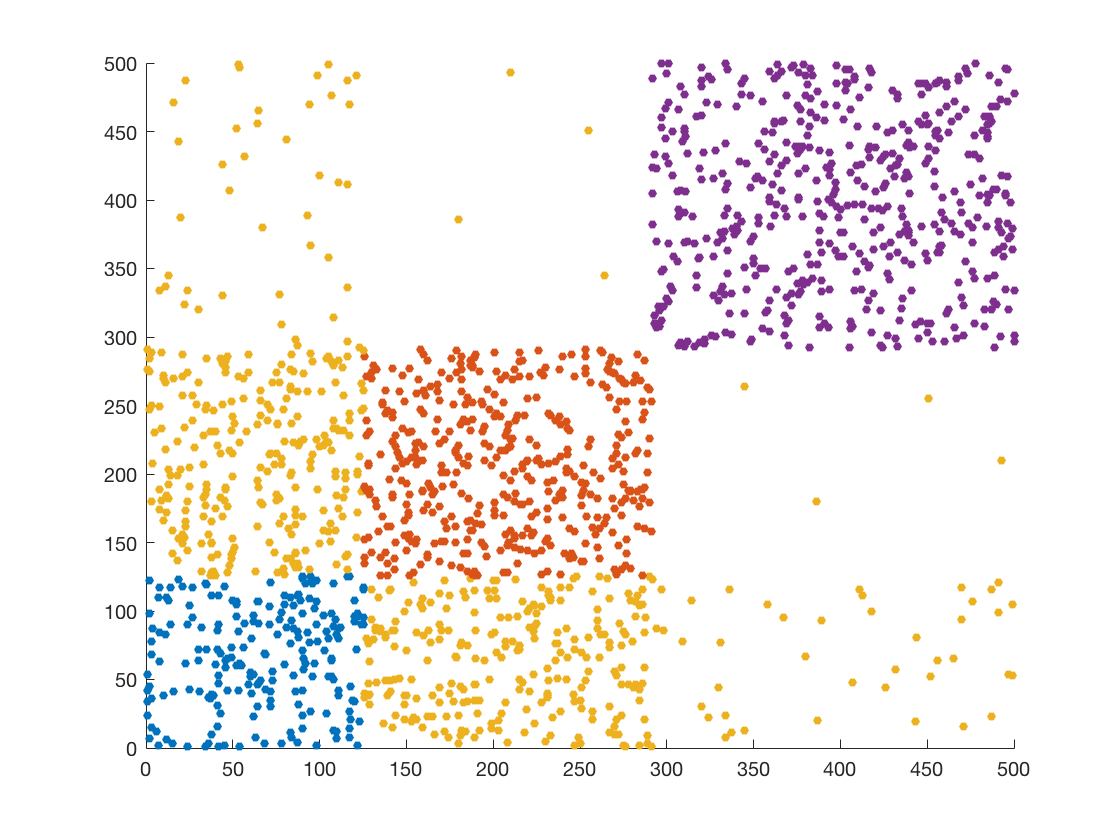}&
\includegraphics[width=0.35\textwidth]{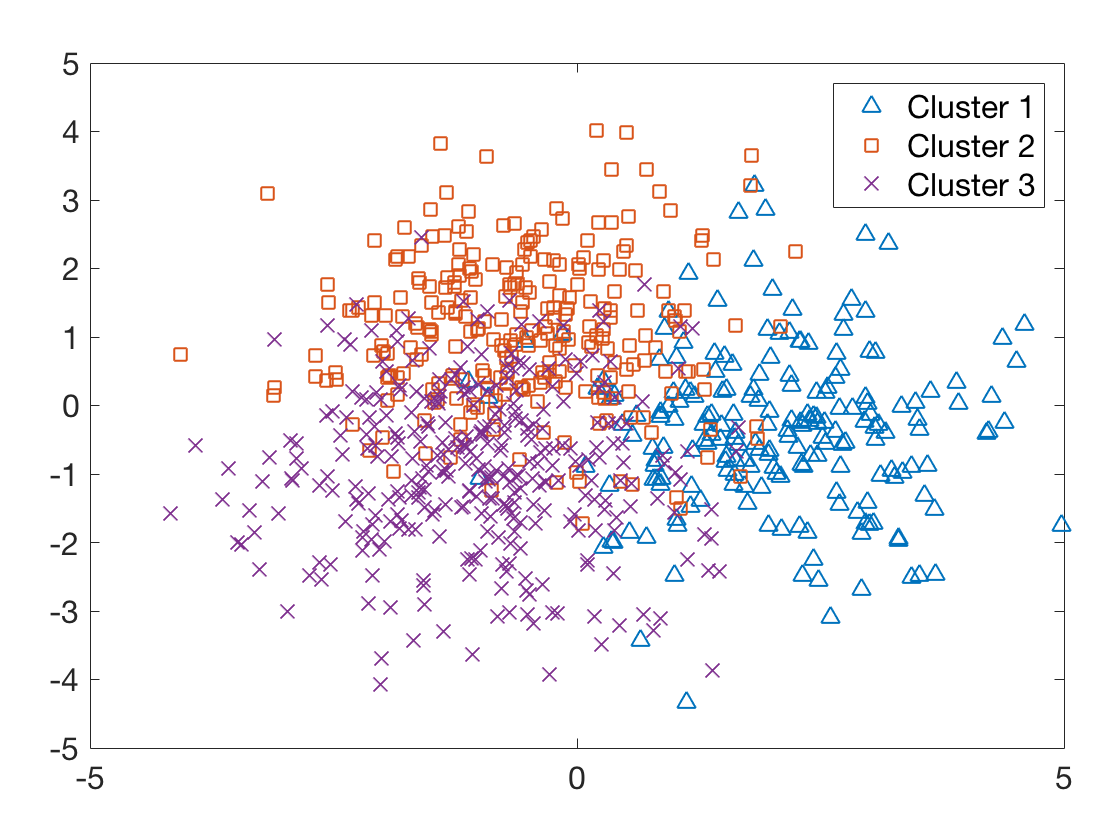}&
\includegraphics[width=0.35\textwidth]{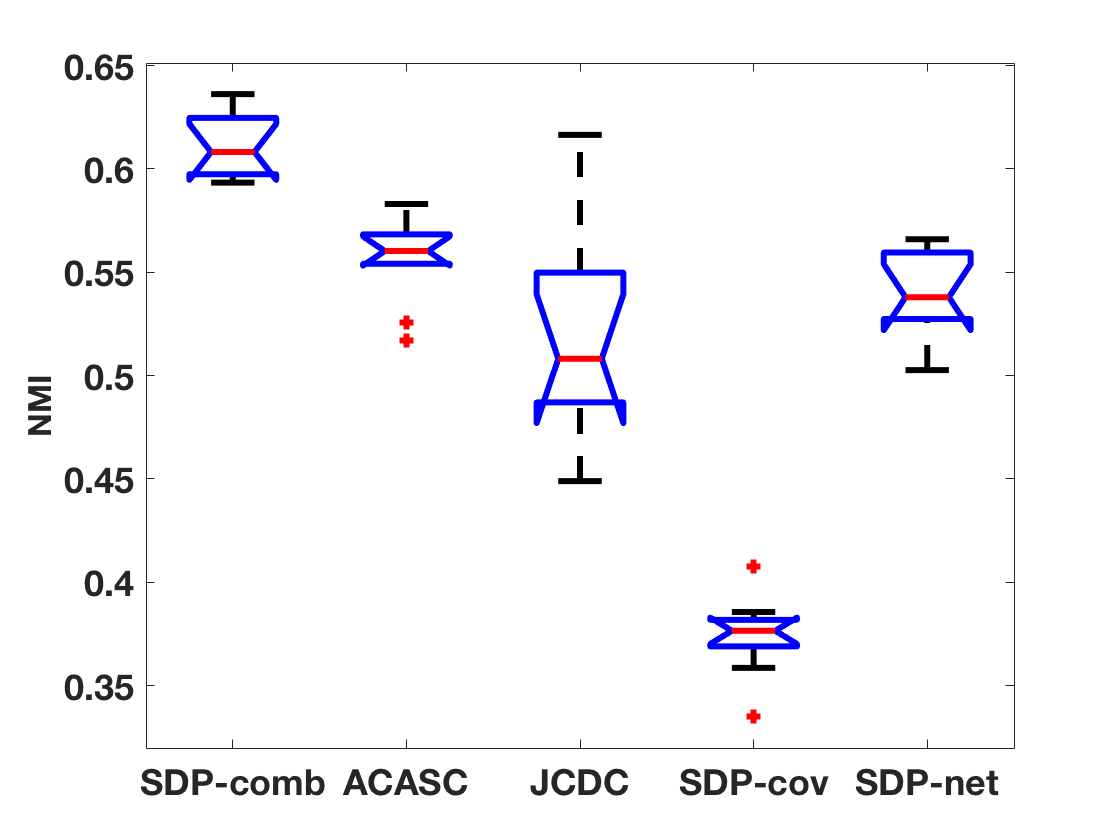}\\
Simulation 1: (a) Graph & (b) Covariates - 1 & (c) NMI - 1\\
\includegraphics[width=0.35\textwidth]{orthogonal_scatter_color.png}&
\includegraphics[width=0.35\textwidth]{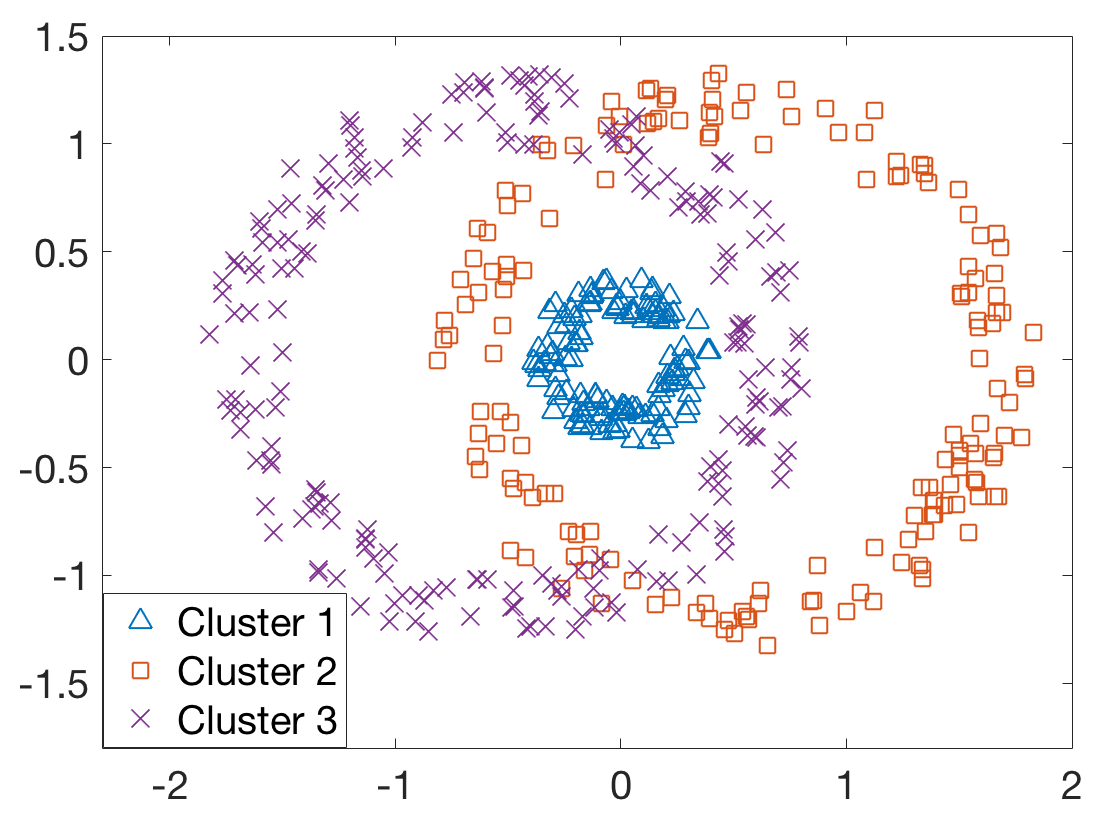}&
\includegraphics[width=0.35\textwidth]{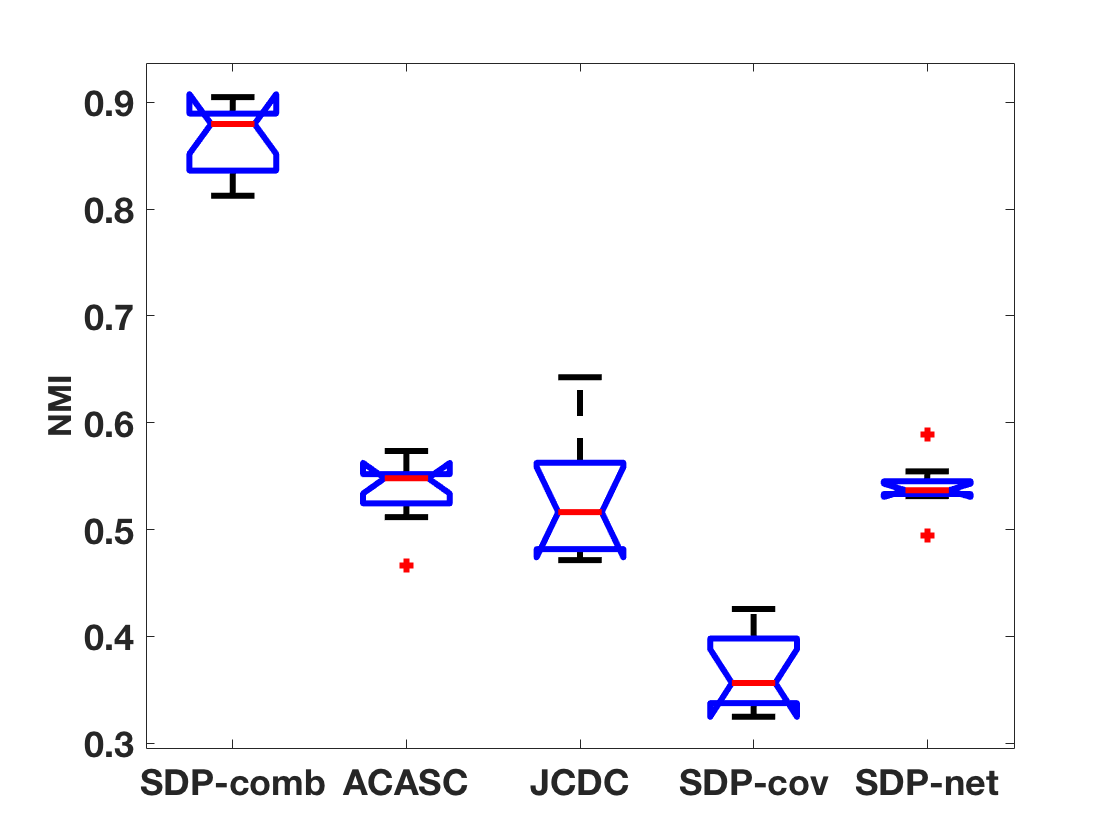}\\
Simulation 2: (d) Graph & (e) Covariates - 2 & (f) NMI - 2
\end{tabular}
\caption{The first and second rows have results for isotropic Gaussian covariates and covariates lies on a nonlinear manifold respectively.  We plot the adjacency matrix $A$ in (a) and (b), where blue, red and purple points represent within cluster edges for 3 ground truth clusters respectively and yellow points represent inter-cluster edges. In (b) and (e) we plot covariates ; different shapes and colors imply different clusters. (c) and (f) show the box plots for NMI. }
\label{fig:simu1}
\end{figure}
In the second row of Figure~\ref{fig:simu1}, we examine covariates with nonlinear cluster boundaries. The graph used here is the same as above, and the covariates are 2-dimensional, whose scatter plot is shown in Figure~\ref{fig:simu1}(e). In this case, the kernel matrix is able to pick up local similarities hence performs better than combination via inner product similarity as used in ACASC. In both simulations, SDP-comb outperforms others.

\subsection{Real World Networks}
Now we present results on a real world social network and an ecological network. 
The performance of clustering is evaluated by NMI with the ground truth labels. 
\begin{figure}[!ht]
\centering
\begin{tabular}{ccc}
\hspace{-2pt}\includegraphics[width=.29\textwidth]{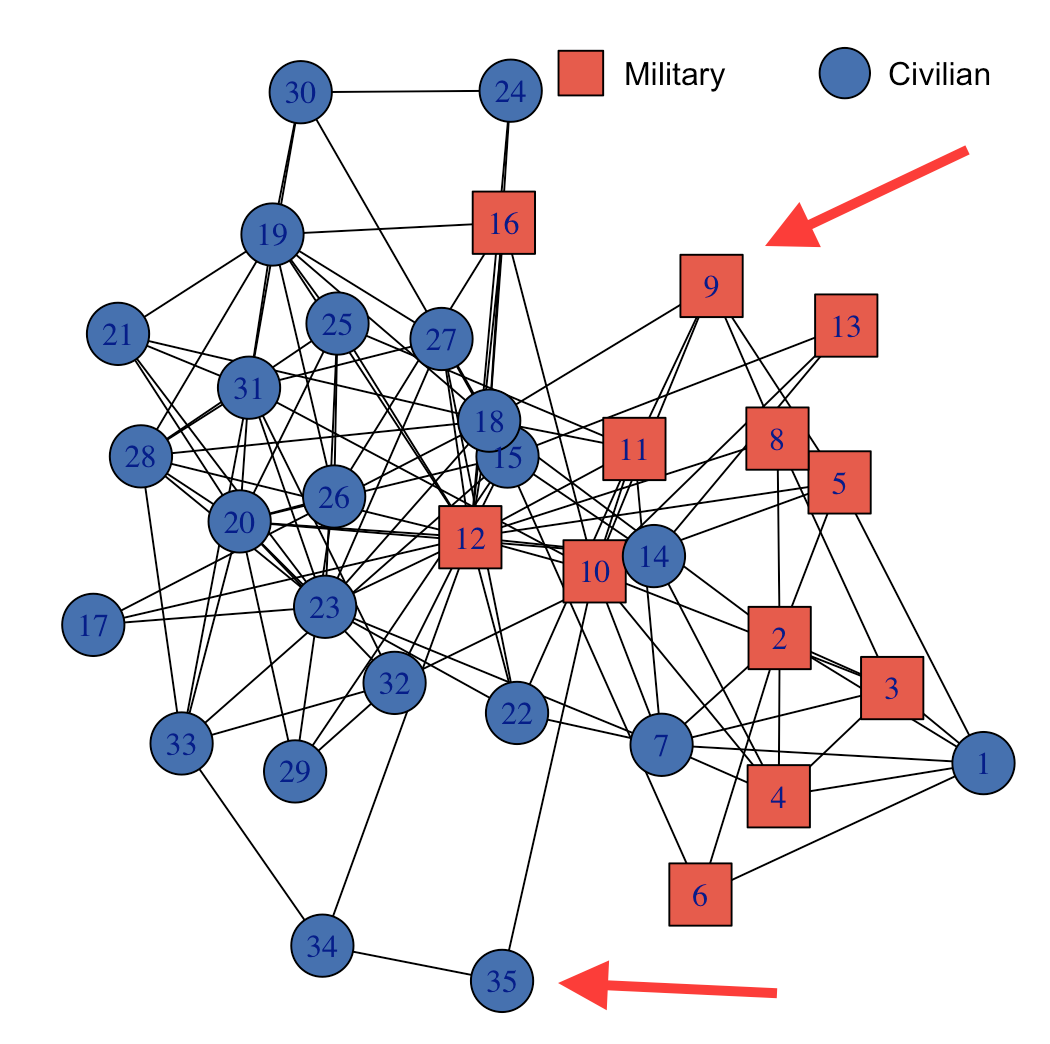} &\hspace{-5pt}
\includegraphics[width=.35\textwidth]{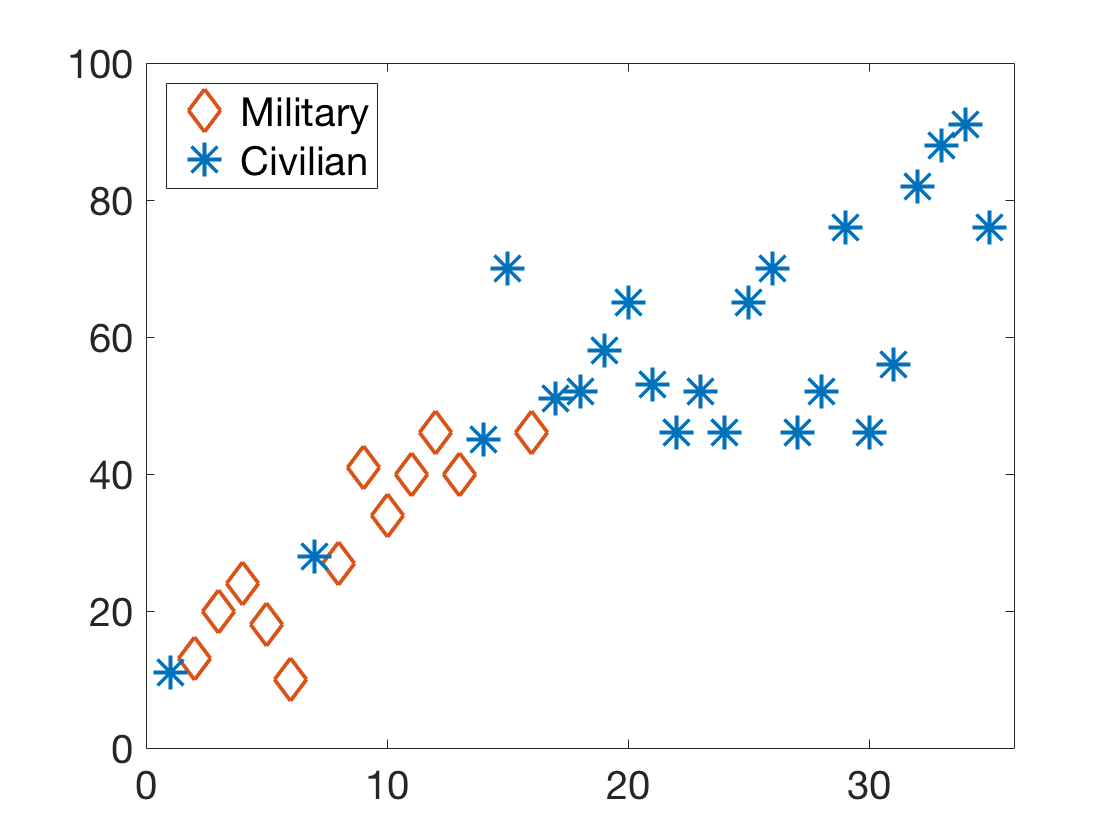} &\hspace{-4pt}
\hspace{-2pt}\includegraphics[width=.29\textwidth]{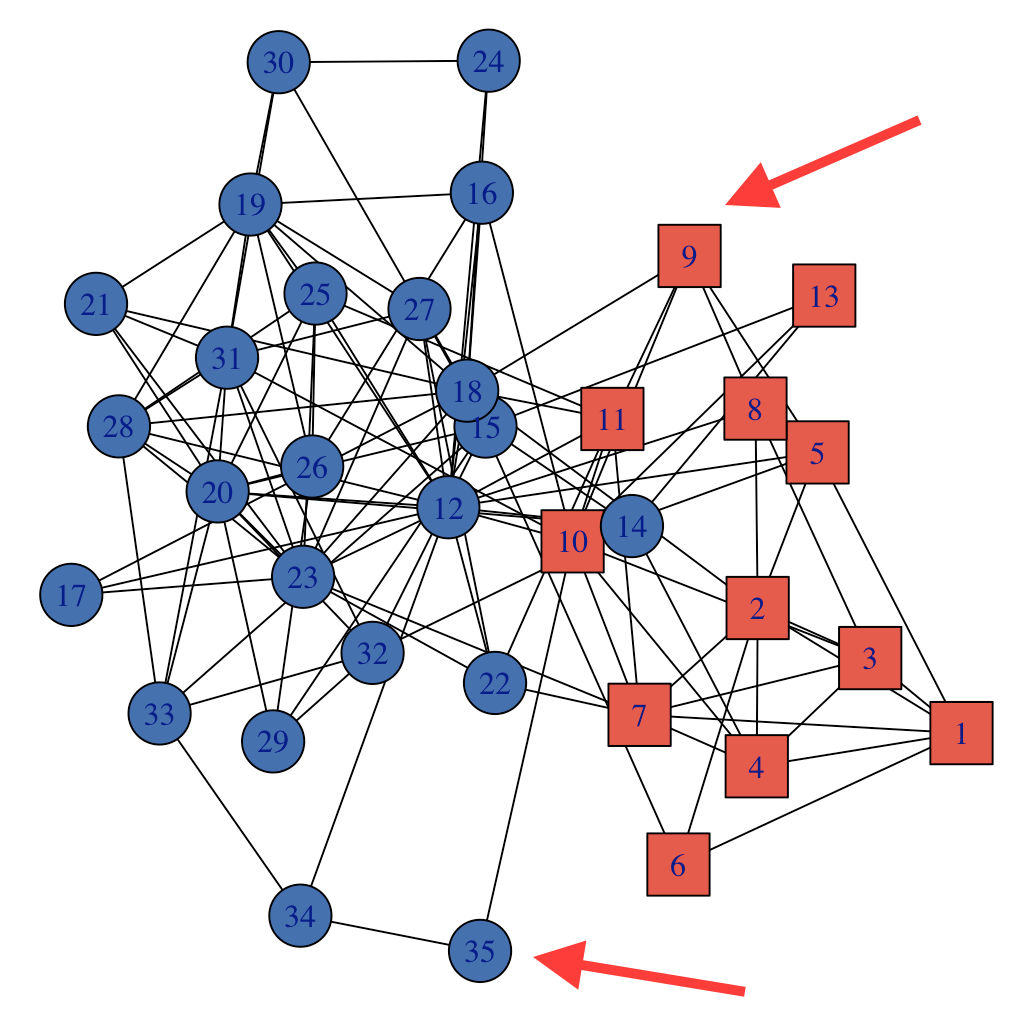}\\
(a) Ground truth & (b) Node feature &\hspace{-1pt} (c) Predicted by SDP-comb
\end{tabular}
\caption{Mexican political network.
}
\label{fig:mex_year}
\end{figure}

\paragraph{Mexican political elites}
As discussed before, this network \citep{gil1996political} depicts the political, kinship, or business interactions between 35 Mexican presidents and close collaborators, etc. The two ground truth clusters consist of the military and the civilians, indicating the background of the politician. The year in which a politician first held a significant governmental position, is used as a covariate. Figure \ref{fig:mex_year}(b) shows that the covariate gives a good indication of the labels. This is because the military dominated the political arena after the revolution in the beginning of the twentieth century, and were succeeded by the civilians.

Table \ref{tab:real} shows the NMI of all methods, where our method outperforms other covariate-assisted approaches. From Figure \ref{fig:mex_year}(a, c), for example, node 35 has exactly one connection to each of the military and civilian groups, but seized power in the 90s, which strongly indicates  a civilian background. On the other hand, node 9 took power in 1940, a year when civilian and military had almost equal presence in politics, making it hard to detect node 9's political affiliation. However, this node has more edges to the military group than the civilian group. By taking the graph structure into consideration, we can correctly assign the military label to it.
\begin{figure}[!h]
\centering
\begin{tabular}{ccc}
\hspace{-15pt} 
\includegraphics[width=.35\textwidth, height=3.8cm]{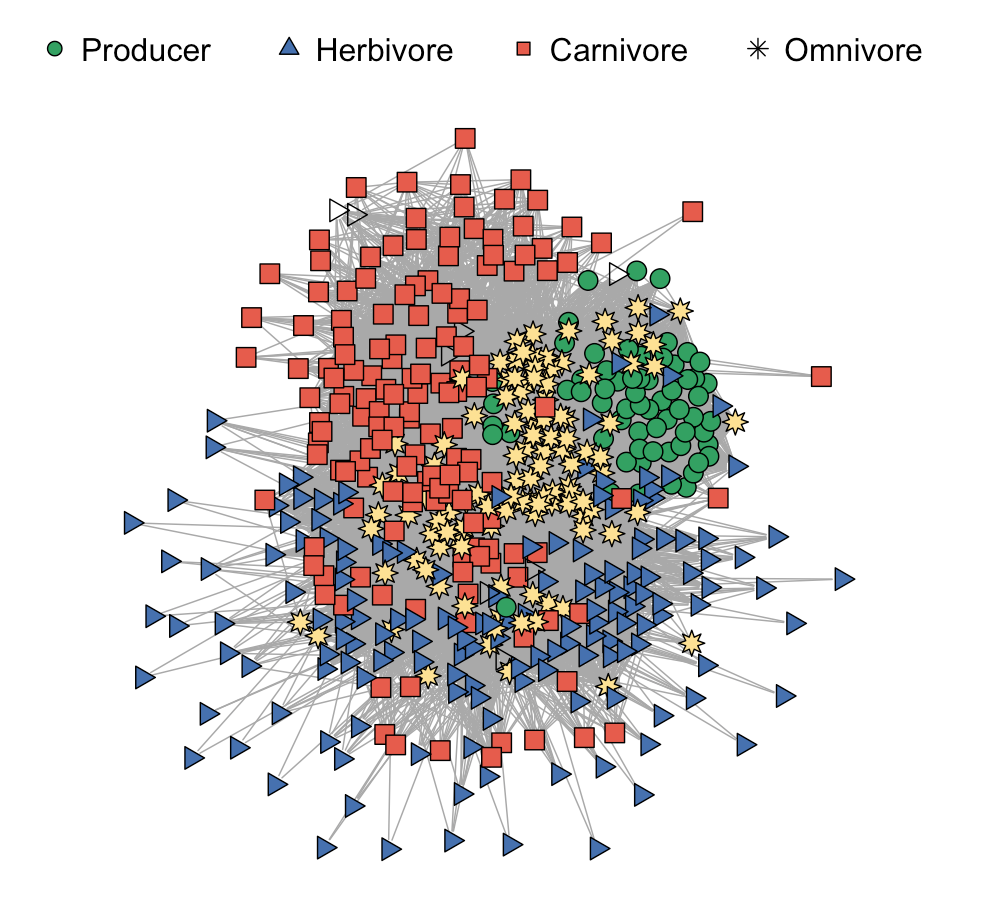}& 
\includegraphics[width=.32\textwidth, height=3.7cm]{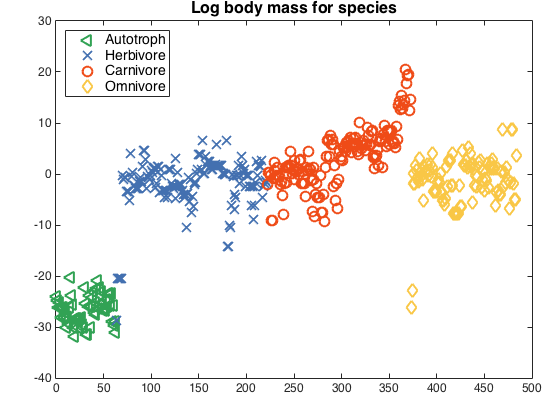} &
\hspace{0pt}\includegraphics[width=.3\textwidth, height=3.7 cm]{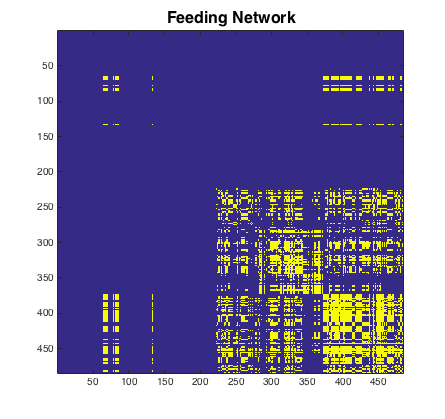} \\
(a)&(b)&(c)\\
\includegraphics[width=.31\textwidth, height=3.7 cm]{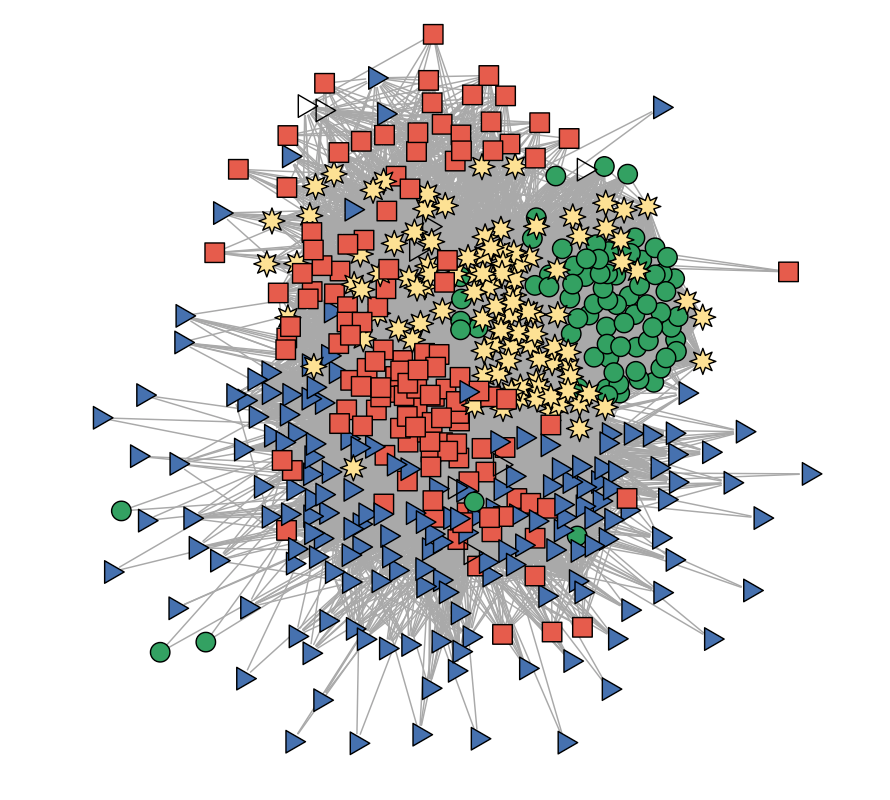} & \hspace{-15 pt}
\includegraphics[width=.31\textwidth, height=3.7 cm]{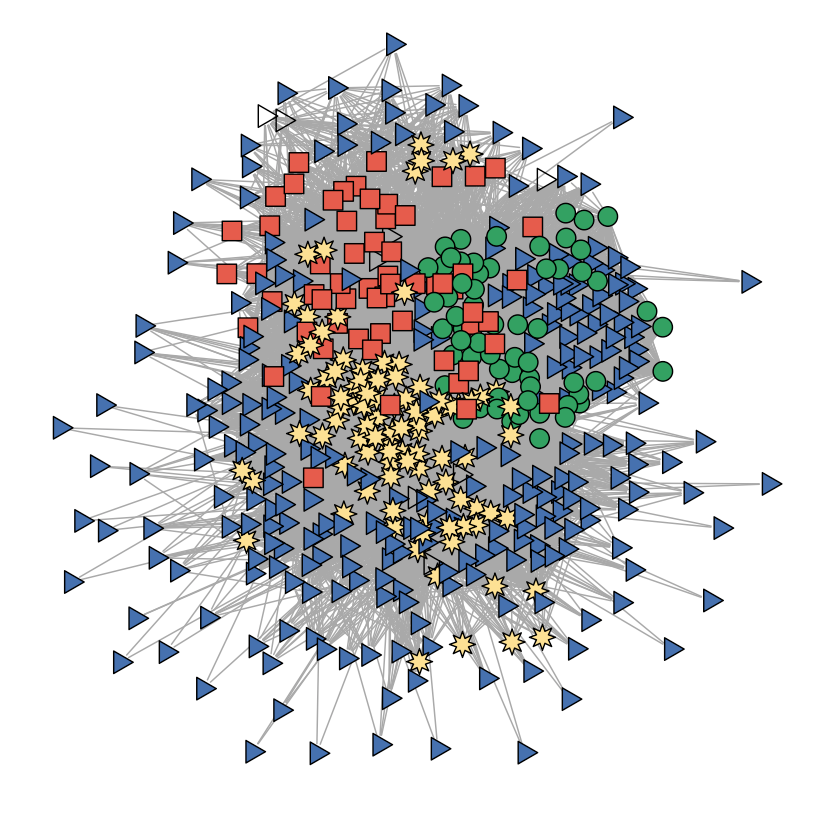}& \hspace{-15 pt}
\includegraphics[width=.31\textwidth, height=3.7 cm]{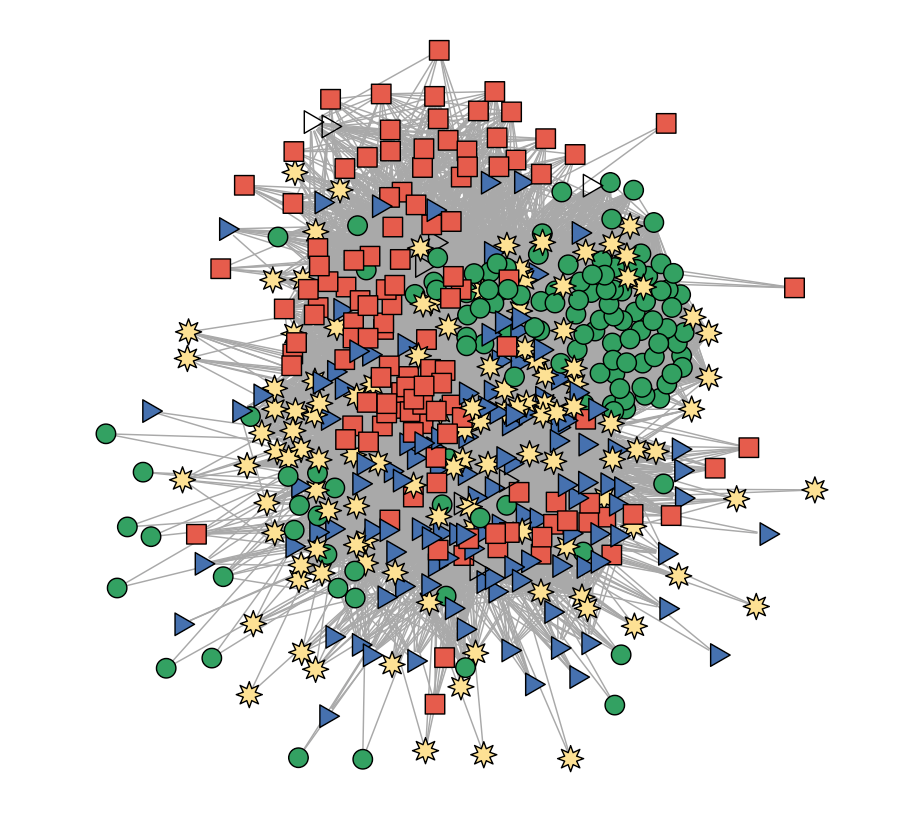}\\
(d) & (e) & (f) 
\end{tabular}
\caption{Weddell sea network: (a) True labels;  (b) Log body mass; (c) Constructed adjacency matrix $A_\tau$; we show labels from (d) SDP-comb; (e) SDP-net; (f) SDP-cov.}
\label{fig:log_body_mass}
\end{figure}

\paragraph{Weddell sea trophic dataset}

The next example we consider is an ecological network collected by \cite{jacob2011role} describing the marine ecosystem of Weddell Sea, a large bay off the coast of Antarctica. The dataset lists 489 marine species and their directed predator-prey interactions, as well as the average adult body mass for each of the species. We use a thresholded symmetrization of the directed graph as the adjacency matrix. Let $G$ be the directed graph, the $(i,j)^{th}$ entry of $GG^T$ captures the number of other species which $i$ and $j$ both feed on. We create binary matrices $A_\tau=1(GG^T\geq \tau)$. Choosing different $\tau$'s between 1 to 10 gives similar clustering. We use $\tau=5$.

All species are labeled into four categories based on their prey types. Autotrophs (e.g. plants) do not feed on anything. Herbivores feed on autotrophs. Carnivores feed on animals that are not autotrophs, and the remaining are omnivores, which feed both on autotrophs and other animals (herbivore, carnivore, or omnivores). Since body masses of species vary largely from nanograms to tons, we work with the normalized logarithm of mass following the convention in~\citet{newman2015structure}. Figure \ref{fig:log_body_mass}(b) illustrates the log body mass for species. Without loss of generality, we order the nodes as autotrophs, herbivores, carnivores and omnivores.

In Figures \ref{fig:log_body_mass}(c), we plot $A_\tau$. Since the autotrophs do not feed on other species in this dataset, and since herbivores do not have too much overlap in the autotrophs they feed on, the upper left corner of the input network is extremely sparse. On the other side, the body sizes for autotrophs are much smaller than those of other prey types. Therefore the kernel matrix clearly separates them out. 

\begin{table}[]
\centering
\begin{tabular}{l|ccccc}
Dataset   & SDP-net  & SDP-cov & SDP-comb   &  ACASC  & JCDC  \\
\hline
Mexican politicians & 0.37 & 0.43 & \bf{0.46} & 0.37 & 0.25\\
\hline
Weddell Sea & 0.36 & 0.22 & \bf{0.51}	& 0.32 & 0.42\\
\hline
\end{tabular}
\caption{\label{tab:real} NMI with ground truth for various methods}
\end{table}
We see that SDP-net (Figure \ref{fig:log_body_mass}(e)) heavily misclusters the autotrophs since it only replies on the network. SDP-net (Figure \ref{fig:log_body_mass}(f)) only takes the covariates into account and cannot distinguish herbivores from omnivores, since they possess similar body masses. However, SDP-comb (Figure \ref{fig:log_body_mass}(d)) achieves a significantly better NMI by combining both sources. 
Table \ref{tab:real} shows the NMI  between predicted labels and the ground truth from SDP-comb, JCDC and ACASC. While JCDC and ACASC can only get as good as the the best of graph or covariates, our method achieves a higher NMI.

\section{Discussion}
 In this paper, we propose a regularized convex optimization framework to infer community memberships jointly from sparse networks and finite dimensional covariates. We theoretically show that our framework can improve clustering accuracy of either source under weaker separation conditions. In particular, when each source only has partial information about the clustering, our methodology can lead to high clustering accuracy, when either source fails.
 We demonstrate the performance of our methodology on simulated and real networks, and show that it in general performs better than other state-of-the-art methods. While for ease of exposition we limit ourselves to two sources, our method can be easily generalized to multiple views or sources. Empirically, we demonstrate that our method works for covariates with nonlinear cluster  boundaries; we intend to extend our theoretical analysis to this setting and non-isotropic covariates as well.

\section*{Acknowledgements}
We thank Arash Amini and Yuan Zhang for generously sharing their code. We are grateful to Soumendu Mukherjee, Peter J. Bickel, Dave Choi and Harry Zhou for interesting discussions on our paper.

\bibliographystyle{agsm}
\bibliography{Bibliography}

\appendix

\section{Background materials on sub-gaussian random vectors}
\label{sec:proof_kernel}
In this section, we present some properties of sub-gaussian random variables. A sub-gaussian random variable is defined by the following equivalent properties. More discussions on this topic can be found in \cite{vershynin2010introduction}.

\begin{lemma}[\cite{vershynin2010introduction}]
The sub-gaussian norm of $X$ is denoted by $\|X\|_{\psi_2}=\sup_{p\ge 1} p^{-1/2}(\bE |X|^p)^{1/p}$. A random vector $X\in \bR^n$ is defined to be sub-gaussian if the one-dimensional marginals $\innerprod{X}{x}$ are sub-gaussian random variables for all $x\in \bR^n$ with sub-gaussian norm $\|X\|_{\psi_2} = \sup_{x\in S^{n-1}}\|\innerprod{X}{x}\|_{\psi_2}$. 

Every sub-gaussian random variable $X$ satisfies:
\begin{itemize}
\item[(1)] $P(|X|>t) \le \exp(1-ct^2/\|X\|^2_{\psi_2})$ for all $t\ge 0$;
\item[(2)] (Rotation invariance) Consider a finite number of independent centered sub-gaussian random variables $X_i$. Then $\sum_i X_i$ is also a centered sub-gaussian random variable. Moreover, $\|\sum_i X_i\|_{\psi_2}^2\le C\sum_i\|X_i\|_{\psi_2}^2 $.
\item[(3)] Let $X_1,\cdots, X_n$ be independent centered sub-gaussian random variables. Then $X=(X_1,X_2,\cdots, X_n)$ is a centered sub-gaussian random vector in $\bR^n$ and $\|X\|_{\psi_2} \le C\max_i \|X_i\|_{\psi_2}$.
\end{itemize}
\label{lem:subgaussian}
\end{lemma}

A random variable is sub-exponential if the following equivalent properties hold with parameters $K_i>0$ differing from each other by at most an absolute constant factor:
(1) $P(|X|>t)\le \exp(1-t/K_1)\ \text{ for all } t\ge 0$; (2) $(\bE|X|)^{1/p} \le K_2p\ \text{ for all }p\ge 1$; (3) $\bE \exp(X/K_3)\le e$. The square of sub-gaussian random variable is sub-exponential.

\begin{lemma}[\cite{vershynin2010introduction}]
A random variable $X$ is sub-gaussian if and only if $X^2$ is sub-exponential. Moreover, $ \|X\|^2_{\psi_2}\le \|X^2\|_{\psi_1} \le2\|X\|_{\psi_2}^2 $.
\label{lem:subexp_sg2}
\end{lemma}

\section{Proof of Lemma~\ref{lem:fro_to_innerprod}}
We start with the following lemma, whose proof is in the Supplementary.
\begin{lemma}
For any $X$ that satisfies $X\succeq 0, X\ge 0, X\bfone=\bfone$, we have $\|X\|_F^2\le \tr(X)$.
\label{lem:fro_less_trace}
\end{lemma}
\begin{proof}
	We first show that for all such $X$, the eigenvalues of $X$ are in $[0,1]$. Let $v_i$ be the eigenvector of $X$ corresponding to the $i^{th}$ largest eigenvalue $\eig_i$. Since $X$ is positive semi-definite,  $\eig_i \ge 0,\forall i$. Without loss of generality, let $i^*=\argmax_i|v_1(i)|$, i.e. be the index of the entry with the largest absolute value of $v_1$. Since $Xv_1=\eig_1 v_1$, 
	and $\sum_j X_{ij}=1, X_{ij}\ge 0$, we have: 
	\bas{
		|\eig_1 v_1(i^*)|=|\sum_j X_{i^*j}v_1(j)|\le\sum_j X_{i^*j}|v_1(j)|\le |v_1(i^*)|.
	} Therefore $|\eig_1| \le 1$.
	\bas{
		\|X\|_F^2=\sum_i \eig_i^2\le \sum_i \eig_i =\tr(X)
	}
\end{proof}
Now we are in position to prove Lemma~\ref{lem:fro_to_innerprod}.
\begin{proof}[Proof of Lemma~\ref{lem:fro_to_innerprod}]
Note that both $X_0$ and \sdp{M} are in the feasible set $\cF$, by optimality, we have $\innerprod{M}{\sdp{M}}\ge \innerprod{M}{X_0} $.
We construct $Q$ as stated in the lemma to obtain: $\innerprod{Q}{\sdp{M}-X_0}$,
 $\innerprod{M-Q}{\sdp{M}-X_0}\ge \innerprod{Q}{X_0-\sdp{M}}$.
Note that $Q$ is constant on diagonal blocks and upper bounded by $q_k$ on off-diagonal blocks, with respect to the clustering of nodes. Using the fact that $|C_k|=m_k$, we have:
\bas{
& \innerprod{M}{X_0-\sdp{M}} = \sum_k \sum_{i\in C_k} \left( \bkin\sum_{j\in C_k}\left( \frac{1}{m_k}-(X_M)_{ij} \right)+\sum_{\ell \ne k}  \sum_{j\in C_\ell} Q_{ij}(0-(\sdp{M})_{ij}) \right)\\
\ge & \sum_k \sum_{i\in C_k} \left( \bkin\sum_{j\in C_k}\left( \frac{1}{m_k}-(\sdp{M})_{ij} \right)- \bkout \sum_{\ell \ne k}\sum_{j\in C_\ell}(\sdp{M})_{ij} \right)\\
= & \sum_k \sum_{i\in C_k} \left( \bkin \left(1-\sum_{j\in C_k}(\sdp{M})_{ij} \right)- \bkout \left( 1-\sum_{j\in C_k}(\sdp{M})_{ij} \right) \right)\\
= & \sum_k \sum_{i\in C_k}(\bkin - \bkout) \left(1-\sum_{j\in C_k}(\sdp{M})_{ij} \right) \ge  \min_k (\bkin - \bkout)  \sum_k \sum_{i\in C_k} \left(1-\sum_{j\in C_k}(\sdp{M})_{ij} \right)
}

The third line and last inequality uses the constraint that $ \sum_{j} \hat{X}_{ij}=1$, and $	1-\sum_{j\in C_k} \hat{X}_{ij} \ge 1-\sum_{j} \hat{X}_{ij}=0 $.
On the other hand,
\bas{
\|\sdp{M}-X_0\|_F^2 
=& \|\sdp{M}\|_F^2-\|X_0\|_F^2+2\innerprod{X_0-\sdp{M}}{X_0}
}
By Lemma~\ref{lem:fro_less_trace}, and the fact that $\|X_0\|_F^2=r$, we have $\|\sdp{M}\|_F^2-\|X_0\|_F^2\le \tr(\sdp{M})-\cl=0$. Since $\min_k (\bkin - \bkout)\geq 0$,
\bas{
&\|\sdp{M} -X_0\|_F^2 \le  2\innerprod{X_0-\sdp{M}}{X_0} = 2\sum_k\sum_{i\in C_k}  \sum_{j\in C_k} \frac{1}{m_k}\left( \frac{1}{m_k}-(\sdp{M})_{ij} \right)\\
=& 2\sum_k\sum_{i\in C_k} \frac{1}{m_k} \left( 1-\sum_{j\in C_k} (\sdp{M})_{ij} \right)\le  \frac{2}{\mmin}\sum_k\sum_{i\in C_k} \left( 1-\sum_{j\in C_k} (\sdp{M})_{ij} \right)\\
\le & \frac{2}{\mmin \min_k (\bkin-\bkout)} \innerprod{Q}{X_0-\sdp{M}}\le \frac{2}{\mmin \min_k (\bkin-\bkout)}  \innerprod{M-Q}{\sdp{M}-X_0}
}
\end{proof}

\section{Proof of Proposition~\ref{prop:main_graph_sparse}}
We first introduce the following result on sparse graph with Grothendieck's inequality by \citet{guedon2014community}.
\begin{lemma}[\cite{guedon2014community}]
Let $\cM_G^+=\{X: X\succeq 0,diag(X)\preceq I_n\}$, $A=(a_{ij})\in \bR^{n\times n}$ be a symmetric matrix whose diagonal entries equal 0, and entries above the diagonal are independent random variables satisfying $0\le a_{ij}\le 1$. Let $P=E[ A|Z]$. Assume that $\bar{p}:=\frac{2}{n(n-1)}\sum_{i<j}\text{Var}(a_{ij})\ge \frac{9}{n}$. Then, with probability at least $1-e^35^{-n}$, we have
$ \max_{X\in \cM_G^+}|\innerprod{A-P}{X}|\le K_G \llnorm{A-P} \le 3K_G\bar{p}^{1/2}n^{3/2},$
where $K_G$ is the Grothendieck's constant, and its best know upper bound is 1.783.
\label{lem:gro_ineq}
\end{lemma}
\begin{proof}[Proof of Proposition~\ref{prop:main_graph_sparse}]
Notice that $A$ and $P:=E[A|Z]$ has zero diagonals. Therefore,
\beq{
\bsplt{
\innerprod{P-Q}{\sdp{A}-X_0} =& \sum_k \sum_{i\in C_k} a_k/n \left( \frac{1}{m_k}-(\sdp{A})_{ii} \right)\\
\le & \sum_k p_k-\pmin \tr(\sdp{A})\le \cl(\pmax-\pmin)},
\label{eq:dense_p-q}
}
where $\pmax=\max_k a_k/n$ and $\pmin=\min_k a_k/n$.
Thus by Lemma \ref{lem:fro_to_innerprod} and Eq~\eqref{eq:dense_p-q},
\bas{
\|\sdp{A}-X_0\|_F^2 \le \frac{2}{\mmin \min_k (a_k/n-b_k/n)} \left( \innerprod{A-P}{\sdp{A}-X_0}+\cl(\pmax-\pmin) \right)
}
In sparse regime, both $m_{\min}X_0$ and $m_{\min}\sdp{A}$ belong to the set $\cM_G^+$. Let  $g=n\bar{p}\geq 9$, applying Lemma~\ref{lem:gro_ineq} we get with probability at least $1-e^35^{-n}$, 
\begin{align*}
\| \sdp{A} - X_0\|_F^2 
\le & \frac{22 \sqrt{n^2g}}{\mmin^2 \min_k(a_k/n-b_k/n)}+ \frac{2\cl (p_{\max}-p_{\min})}{\mmin\min_k(a_k/n-b_k/n)}
\end{align*}
Substituting $p_k=a_k/n, q_k=b_k/n$, and using the fact that 
\bas{
	\frac{2\cl (p_{\max}-p_{\min})}{\mmin\min_k(p_k-q_k)}&=	\frac{2\cl\mmin (p_{\max}-p_{\min})}{\mmin^2\min_k(p_k-q_k)}
	\leq \frac{2 \max_k a_k }{\mmin^2\min_k(p_k-q_k)}=o(\sqrt{n^2g}),
} 
Recall that $\alpha:=m_{\max}/m_{\min}$, we get with probability tending to 1, 
\bas{
\frac{\|\hat{X}-X_0\|_F^2}{\|X_0\|_F^2} \le  \frac{ 23 n^2 \sqrt{g}}{r\mmin^2\min_k (a_k-b_k)}\leq \frac{ 23 \alpha^2\cl \sqrt{g}}{\min_k (a_k-b_k)}.
}
\end{proof}

\section{Proof of Proposition~\ref{prop:lowdim_kernel}}
\label{sec:proof_sparse}

\begin{proof}[Proof of Proposition~\ref{prop:lowdim_kernel}]
	Recall that by definition, for $i\in C_k$, $Y_i-\mu_k$ is sub-gaussian random vector with sub-gaussian norm $\psi_k$.
	Using the following concentration inequality from \cite{hsu2012tail} for sub-gaussian random vectors, we have:
	\begin{align*}
	\mbox{For $i\in C_k$, }P(\|Y_i-\mu_k\|_2^2>\psi_k^2(d+2\sqrt{t\dim}+2t))\le e^{-t}
	\end{align*}	
	We take  $t = c_k^2 d$ for $c_k\geq 1$. Since $1+2c_k+2c_k^2\leq 5c_k^2$ for $c_k\geq 1$, we get $P(\|X-\bE X\|^2 \le 5c_k^2\psi_k^2d)\ge 1-\exp(-c_k^2d)$. Let $\Delta_k = \sqrt{5}c_k\psi_k\sqrt{d}$, we can divide the nodes into ``good nodes'' (those close to their population mean) $\cS_k$ and the rest as follows:
	\begin{align}
	\label{eq:csk}
	\cS_k=\{i\in C_k:\|Y_i-\mu_k\|\leq \Delta_k\},\qquad \cS=\cup_{k=1}^\cl \cS_k
	\end{align} 
	
	Let $m_c^{(k)}=m_k-|\cS_k|$. We want to bound $m_c^{(k)}$ with high probability. Note that $m_c^{(k)}=\sum_{i\in C_k} \bm{1}(\|Y_i-\mu_k\|\geq \Delta_k)$ is a sum of i.i.d random variables. Therefore, using the Hoeffding bound we have:
	\bas{
		P\left(m_c^{(k)}-m_k P(i\not\in \cS_k)\geq m_k\delta\right)\leq \exp(-2m_k\delta^2)
	}
	Using $\delta=\sqrt{\log m_k/2m_k}$,
	we have:
	\bas{
		P\left(m_c^{(k)}-m_k P(i\not\in \cS_k)\geq \sqrt{m_k\log m_k/2}\right)\leq \frac{1}{m_k}
	}
	Since $P(i\not\in \cS_k)\leq \exp(-c_k^2d)$, we have:
	\bas{
		P\left(m_c^{(k)}\geq m_k \exp(-c_k^2d)+ \sqrt{m_k\log m_k/2})\right)\leq \frac{1}{m_k}
	}
	Finally, using union bound over all clusters we get:
	\begin{align}
	P\left(m_c \ge \sum_k m_k e^{-c_k^2d}+ \sum_k\sqrt{m_k\log m_k/2}\right)\le \sum_k \frac{1}{m_k}
	\label{eq:chernoff_noutlier}
	\end{align}
	Now define 
	\beq{
		(K_I)_{ij}=\left\{ \begin{matrix} f(2\Delta_k), & \mbox{ if $i,j\in C_k$}\\ \min\{f(d_{k\ell}-\Delta_k-\Delta_\ell),K_{ij}\}, & \mbox{ if $i\in C_k, j\in C_\ell, k\ne \ell$}\end{matrix}\right. 
		\label{eq:K_I}
	}
	By Lemma \ref{lem:fro_to_innerprod}, all diagonal blocks are blockwise constant and the off-diagonal blocks are upper bounded by $f(d_{k\ell}-\Delta_k-\Delta_\ell)$. Let $\nu_k = f(2\Delta_k)-\max_{\ell\ne k}f(d_{k\ell}-\Delta_k-\Delta_\ell)$, and $\gamma=\min_k \nu_k$. If $\nu_k\geq 0$, we have
	\bas{
		\| \sdp{K}-X_0\|_F^2 \le \frac{2}{\mmin \gamma}\innerprod{K-K_I}{ \sdp{K} -X_0}
	}
	Apply Grothendieck's inequality,
	\ba{
		\| \sdp{K} -X_0\|_F^2 \le \frac{2K_G}{\mmin^2 \gamma}\llnorm{K-K_I}
		\label{eq:lowdim_ki}
	}
	Now it remains to bound the $\ell_\infty\to\ell_1$ norm of $K-K_I$. Note that if $i\in S_k, j\in S_\ell, k\ne\ell$, then by a simple use of triangle inequality we have $K_{ij}\le f(d_{k\ell}-\Delta_k-\Delta_\ell)$, so $K_{ij}=(K_I)_{ij}$; and if $i,j\in S_k$, then $K_{ij}\ge f(2\Delta_k)$.
	\beq{
		\bsplt{
			\llnorm{K-K_I} =& \max_{x,y\in \{\pm\}^n} \sum_{i,j} x_iy_j\left( K_{ij}-(K_I)_{ij} \right)\\
			\le& \max_{x,y\in \{\pm\}^n} \sum_{i,j\in \cS} x_iy_j\left( K_{ij}-(K_I)_{ij} \right)+\max_{x,y\in \{\pm\}^n} \sum_{i \not \in \cS \cup j\not \in \cS} x_iy_j\left( K_{ij}-(K_I)_{ij} \right)\\
			\stackrel{(i)}{\le}& \max_{x,y\in \{\pm\}^n} \sum_{i,j\in \cS} x_iy_j\left( K_{ij}-(K_I)_{ij} \right)+2m_cn\\
			\stackrel{(ii)}{=}& \max_{x,y\in \{\pm\}^n} \sum_k \sum_{i,j\in \cS_k} x_iy_j\left( K_{ij}-f(2\Delta_k) \right)+2m_cn\\
			\le & \sum_k m_k^2(1-f(2\Delta_k))+2m_cn
		}
		\label{eq:llnorm_kki}
	}
	where $(i)$ is due to $|K_{ij}-(K_I)_{ij}| \le 1$, and $(ii)$ comes from the definition of $K_I$. Now Eq~\ref{eq:lowdim_ki} follows as
	\beq{
		\bsplt{
			\| \sdp{K} -X_0\|_F^2 \le& \frac{4K_G \left( \sum_k m_k^2(1-f(2\Delta_k))+2m_cn \right)}{\mmin^2 \gamma}\\
			=& \frac{4K_G}{\mmin^2}\sum_k \left(m_k^2\frac{1-f(2\Delta_k)}{\gamma}+2m_kne^{-c_k^2d}/\gamma\right)+\frac{\sqrt{2}K_G n}{\mmin^2\gamma}\sum_k \sqrt{m_k\log m_k} 
		}
		\label{eq:lowdim_sumk}
	}
	
	Recall that $f(x)=\exp(-\eta x^2)$, and $\gamma =\min_k \left\{ f(2\Delta_k)-\max_{\ell\ne k}f(d_{k\ell}-\Delta_k-\Delta_\ell)\right\}$. For simplicity, we assume $c_k=c_0$. We take $c_0 =\sqrt{ \log\left( \frac{\dmin^2}{\psi_{\max}^2d} \right)\left/ d \right.}$ and the scale parameter $\eta = \frac{\phi }{20c_0^2\psi_{\max}^2d}$, for some $\phi>0$, which will be chosen later. Furthermore, we also define 
	\ba{
		\label{eq:xi}
		\xi = \frac{\dmin}{2\sqrt{5}c_0\psi_{\max}\sqrt{d}}-1
		.
	} 
	
	If $\xi>1$, then $d_{min}>4\sqrt{5}c_0\psi_{max}\sqrt{d}$, and hence $\gamma>0$. Also, since $\eta(\dmin-2\sqrt{5}c_0\psi_{\max}\sqrt{d})^2=\phi\xi^2$, 
	$\forall k,\ell\in[\cl]$, if $\dmin:=\min_{k\ell}d_{k\ell}>4\sqrt{5}c_0\psi_{\max}\sqrt{d}$, then 
	\[ \gamma \ge f(2\sqrt{5}c_0\psi_{\max}\sqrt{d})-f(\dmin-2\sqrt{5}c_0\psi_{\max}\sqrt{d})=\exp(-\phi)-\exp(-\phi\xi^2). \]
	and 
	\[ 1-f(2\Delta_k)\le 1-f(2\sqrt{5}c_0\psi_{\max}\sqrt{d})=1-\exp(\phi) \]
	
	Recall $\alpha=\frac{m_{\max}}{\mmin}$, 
	\ba{
		&\| \sdp{K} -X_0\|_F^2 \\
		\le& 4K_Gr\alpha^2 \cdot \frac{1-f(2\sqrt{5}c_0\psi_{\max}\sqrt{d})+2r\exp(-c_0^2d)}{\gamma}+\frac{2\sqrt{2}K_G m_{\max}\cl^2\sqrt{m_{\max}\log m_{\max}} }{\gamma m_{\min}^2}\nonumber\\
		\le & \frac{4K_Gr\alpha^2}{\gamma} \left( 1-\exp(-\phi)+\frac{2r\psi^2_{\max}\sqrt{d}}{\dmin^2}+ \cl\sqrt{\log m_{\max}/2m_{\max}} \right)\nonumber\\
		\le & 4K_Gr\alpha^2\left(\underbrace{\frac{  (1-\exp(-\phi)+2r\psi^2_{\max}d/\dmin^2}{\exp(-\phi)-\exp(-\phi\xi^2)}}_{A}+ \underbrace{\frac{\cl\sqrt{\log m_{\max}/2m_{\max}} }{\exp(-\phi)-\exp(-\phi\xi^2)}}_{B}\right)
		\label{eq:AB}}
	
	We will first bound part (A).
	\ba{
		& (A) =  \frac{\exp(\phi)-1+\exp(\phi)\frac{2r\psi^2_{\max}d}{\dmin^2}}{1-\exp(\phi-\phi \xi^2 )} \stackrel{(i)}{\le}  \frac{\phi+\frac{\phi^2}{2}\exp(\phi)+\exp(\phi)\frac{2r\psi^2_{\max}d}{\dmin^2}}{1-\exp(\phi-\phi \xi^2 )}
		\label{eq:A}
	}
	where $(i)$ uses the Mean value theorem: for $e^x-1\leq x+e^{y}x^2/2$ for $y\in[0,x]$.
	If $\frac{\dmin}{\psi_{\max}\sqrt{d}}>\max\left\{1,\frac{180}{d} \right\}$, using the fact that $\log x\leq \sqrt{x}$, we have: 
	\bas{
	\frac{\dmin^2}{\psi_{\max}^2d}>\frac{180}{d^2}\frac{\dmin}{\psi_{\max}}>\frac{180}{d}\log\left( \frac{\dmin^2}{\psi_{\max}^2d} \right)=180c_0^2.
	}
	Using Eq~\ref{eq:xi}, we see that $\xi>\frac{\sqrt{180}}{2\sqrt{5}}-1=2$, and hence $\gamma>0$. Now we pick $\phi=\frac{\log \xi}{\xi^2}$.
	
	Now we will use this to obtain a lower bound on $1-\exp(\phi-\phi \xi^2 )$. Since $\xi\geq 2$, we have $\xi^2/4\geq 1$. Hence 
	\bas{
		1-\exp(\phi-\phi \xi^2 )&\geq 1-\exp(\phi\xi^2/4-\phi \xi^2 )\\
		&= 1-\exp(-\phi 3\xi^2/4)=1-\exp(-3\log\xi/4)=1-\xi^{-3/4}\\
		&\geq 1-2^{-3/4}=.4
	}
	Using the fact that the function $\frac{\log x}{x^2}$ is monotonically decreasing when $x>2$, we see that $\phi<\log 2/2^2$ and $\exp(\phi)\leq 1.2$. Furthermore,
	\ba{
		\label{eq:gamma-lb}
		\gamma\geq \exp(-\phi)(1-\exp(\phi(1-\xi^2)))\geq .3 
	}
	
	Now Eq.~\eqref{eq:A} yields:
	\bas{
		& (A) \le  \frac{\phi+1.2\left( \frac{\phi^2}{2}+\frac{2r\psi^2_{\max}d}{\dmin^2}\right)}{.4}
		\stackrel{}{\le}  \frac{c\log \xi}{\xi^2}+\frac{3r\psi_{\max}^2d}{\dmin^2} \\
		\stackrel{(ii)}{\le}  & \frac{c'\log (\xi+1)}{(\xi+1)^2}+\frac{3r\psi_{\max}^2d}{\dmin^2}
		\le  c''\frac{\psi_{\max}^2d}{\dmin^2}\log\left(\frac{\dmin}{\psi_{\max}\sqrt{d}}\right)+\frac{3r\psi_{\max}^2d}{\dmin^2},
	}
	for some constant $c$.
	To get $(ii)$, note that 
	\[ \frac{\log \xi}{\xi^2}\le \frac{\log(\xi+1)}{\xi^2}\le \frac{2.25\log(\xi+1)}{(\xi+1)^2}, \forall \xi>2 \]
	
	Finally, we bound (B) in Eq~\ref{eq:AB} using Eq~\ref{eq:gamma-lb}.
	\bas{
		(B)=\frac{\cl\sqrt{\log m_{\max}/2m_{\max}} }{\exp(-\phi)-\exp(-\phi\xi^2)}\leq c_1\cl\sqrt{\frac{\log m_{\max}}{m_{\max}}}
	}
for some constant $c_1>0$.
	Putting pieces together, we have
	\bas{
		\frac{\| \sdp{K} -X_0\|_F^2}{\|X_0\|_F^2} \le& C\alpha^2\max\left( \frac{\psi_{\max}^2d}{\dmin^2} \max\left\{\log\left(\frac{\dmin}{\psi_{\max}\sqrt{d}}\right),r\right\},\cl\sqrt{\frac{\log \mmax}{m_{\max}}}\right)
	}
\end{proof}

\section{Analysis for \sdp{A+\lamn K}}

\begin{proof}[Proof of Theorem~\ref{th:sparse_low}]
Let $K_I$ be defined as in Eq~\eqref{eq:K_I}. Let $\gamma=\min_k(a_k/n-b_k/n+\lamn(f(2\Delta_k)-\max_{\ell\ne k}f(d_{k\ell}-\Delta_k-\Delta_\ell)))$.
When $\gamma\geq 0$, Lemma \ref{lem:fro_to_innerprod} with $Q=ZBZ^T+\lamn K_I$, we have
\bas{
&\| \sdp{A+\lamn K}-X_0\|_F^2 \\
\le& \frac{2}{\mmin \gamma} \left( \innerprod{A-P}{ \sdp{A+\lamn K} -X_0}+\cl(\max_k a_k/n-\min_k a_k/n)+\lamn \innerprod{K-K_I}{ \sdp{A+\lamn K} -X_0} \right)
}
Now by Grothendieck's inequality on both $\innerprod{A-P}{ \sdp{A+\lamn K} -X_0}$ and $\innerprod{K-K_I}{\sdp{A+\lamn K} -X_0}$, one gets,
\bas{
\| \sdp{A+\lamn K} -X_0\|_F^2 \le \frac{2K_G}{\mmin^2 \gamma} \left( 2\llnorm{A-P}+\cl(\max_k a_k/n-\min_k a_k/n) + 2\lamn  \llnorm{K-K_I} \right)
}
By Lemma \ref{lem:gro_ineq} and Eq~\eqref{eq:llnorm_kki},
\bas{
\| \sdp{A+\lamn K} -X_0\|_F^2 \le \frac{4K_G}{\mmin^2 \gamma} \left(6\sqrt{n^3\bar{p}}+\lamn \left( 2m_c n+\sum_{k}m_k^2(1-f(2\Delta_k)) \right) \right)
}

Using $\lamn=\lamz/n$, $m_k=n\pi_k$, $\mmin=n\pi_{\min}$, and $\pi_0:=\sum_k (m_k\exp(-\Delta_k^2/(5\psi_k^2))+ \sqrt{m_k\log m_k/2})/n$ in conjunction with Eq~\eqref{eq:chernoff_noutlier},
we get with probability tending to 1,
\bas{
	\| \sdp{A+\lamn K} -X_0\|_F^2 \le 4K_G\frac{ 6\sqrt{g}+\lamz \left( 2 \pi_0+\sum_{k}\pi_k^2(1-f(2\Delta_k)) \right) }{\pi_{\min}^2 \min_k(a_k-b_k+\lamz \nu_k)}
}
\end{proof}

\section{Analysis of covariate clustering when $\dim \gg \cl$}

Before proving Lemma~\ref{lem:dimreduce}, we clearly state our assumptions and other useful lemmas.
\begin{assumption}
We assume that $M$ is of rank $\cl-1$, i.e. the means are not collinear, or linearly dependent, other than the fact that they are centered.
\label{as:colinear}
\end{assumption}

\begin{lemma}
	\label{lem:cov-conc} 
	Let $M=\sum_k \pi_k \mu_k\mu_k^T$ and $S$ be the covariance matrix of $n$ data points from a sub-gaussian mixture, then $S = M+\sum_i \pi_i\sigma_i^2 I_d$.
	Let $\hat{S}$ be the sample covariance matrix $\hat{S}=\frac{\sum_{i=1}^n (Y_i-\bar{Y}) (Y_i-\bar{Y})^T}{n}$.
We have $\|\hat{S}-S\|\leq C\sqrt{\frac{d\log n}{n}}$ for some constant $C$ with probability bigger than $1-O(n^{-\dim})$. 
\end{lemma}
This is a direct consequence of Corollary 5.50 from~\citet{vershynin2010introduction}.
The main ingredient of the proof is provided below. 

\begin{lemma}
	\label{lem:lt_after_proj}
	Let $U_{\cl -1}$ be the top $\cl -1$ eigenvectors of $\hat{S}$ estimated using $P_1$, and $\lambda$ be the smallest positive eigenvalue of $M$.
	For any vector $v$ in the span of $\{\mu_i\}_{i=1}^\cl$, as long as  $\lambda>5\left(\psi_{\max}^2+C\sqrt{\frac{d\log^2 n}{n}} \right)$ we have $\|U_{\cl-1}^Tv\|\geq \|v\|/2$ with probability at least $1-\tilde{O}(n^{-d})$.
\end{lemma}

\begin{proof}
	Take $n_1=\frac{n}{\log n}$ and $v$ to be a vector in the span of $\{\mu_i\}_{i=1}^\cl$. 
	By definition, we have $\|Mv\|\geq \lambda\|v\|$. Let $R=\hat{S}-S$. Denote $\bar{\sigma}^2 = \sum_i \pi_i \sigma_i^2$, by Lemma~\ref{lem:cov-conc}, $S=M+\bar{\sigma}^2 I_d$. We also know that $\bar{\sigma}^2 \le \sigma_{\max}^2 \le \psi_{\max}^2$ by the property of sub-gaussian distributions.
	Since $S$ is estimated from $P_1$ with $n_1$ points, applying Lemma~\ref{lem:cov-conc} with $n=n_1$ we get $\|R\|\le \epsilon=C\sqrt{\frac{d\log n_1}{n_1}}$.
By Weyl's inequality,  $\|\hat{S}v\|=\|(M+R+\sum_i \sigma_i^2 I_d)v\|\geq (\lambda-\sigma_{\max}^2-\epsilon)\|v\|$.
	Let $U_{\cl:\dim}$ be the eigenspace orthogonal to $U_{\cl-1}$.
	
	Assume the contradiction that $\|U_{\cl-1}^Tv\| < \|v\|/2$. Then there has to be a unit $d$ dimensional vector $u\in \text{span}(U_{\cl:\dim})$, such that $|u^Tv|>\|v\|/2$. 
	On one hand, if we write $u=c\frac{v}{\|v\|}+\sqrt{1-c^2}v^{\perp}$, for $|c|>1/2$ and some unit vector $v^{\perp}$ orthogonal to $v$, we have $\|\hat{S}u\|\geq \frac{\lambda-\sigma_{\max}^2-\epsilon}{2}-\sqrt{1-c^2}\|\hat{S}v^\perp\|$. Note $\|\hat{S}v^\perp\|=\|(M+R+\bar{\sigma}^2 I_d)v^\perp\|$. Since $v^\perp$ is orthogonal to the span of $M$, $\|\hat{S}v^\perp\|\leq (\sigma^2_{\max}+\epsilon) $. Hence 
	\ba{
	\|\hat{S}u\|\geq \frac{\lambda-3(\sigma_{\max}^2+\epsilon)}{2}.
	\label{eq:Sulowerbound}}
	
	 On the other hand, since $u\in \text{span}(U_{\cl:\dim})$, by Weyl's inequality, $\|\hat{S}u\|\leq |\lambda_k(\hat{S})|\leq \sigma_{\max}^2+\epsilon$. This contradicts with Eq.~\eqref{eq:Sulowerbound} since we assume $\lambda>5(\psi_{\max}^2+\epsilon)\ge 5(\sigma_{\max}^2+\epsilon)$. The result is proven by contradiction.
\end{proof}
\begin{remark}
Note that the result can be generalized to non-spherical case as long as the largest eigenvalue of covariance matrix for each cluster is bounded. 
\end{remark}

We are now ready to prove Lemma~\ref{lem:dimreduce}.
\begin{proof}[Proof of Lemma~\ref{lem:dimreduce}]
	Recall that $Y'_i=U_{\cl-1}^T Y_i$ where $U_{\cl-1}$ and $Y_i$ are from two different partitions and hence independent. Let $Z_i\in [\cl]$ denote that latent variable associated with $i$. Thus, $E[Y'_i|Z_i=a,P_2]=U_{\cl-1}^TE[Y_i|Z_i=a]=U_{\cl-1}^T\mu_a$. Thus the means of the new mixture are $\mu'_a:=U_{\cl-1}^T\mu_a$ and the covariance matrix is isotropic, i.e. $E[(Y'_i-\mu'_a)(Y'_i-\mu'_a)^T|P_2,Z_i=a]=\sigma_a^2 I_{r-1}$.  Furthermore, using Lemma~\ref{lem:lt_after_proj} we have
	$\min_{k\neq \ell} \|\mu'_k-\mu'_\ell\|=\min_{k\neq \ell} \|U_{\cl-1}^T (\mu_k-\mu_\ell)\|\geq \|d_{\min}\|/2$. Since this requires an application of Lemma~\ref{lem:lt_after_proj} to each of the vectors $\mu_k-\mu_\ell$, $k,\ell\in[\cl]$, the success probability is at least $1-\tilde{O}(\cl{^2}n^{-d})$ by union bound.
\end{proof}

\section{From $X$ to cluster labels}

From some solution matrix $\xh$, we can apply spectral clustering on it to get the cluster labels. Below we present a theorem that bounds the misclassification error by the Frobenius norm of matrix difference. The proof technique is inspired by those in~\cite{rohe2011spectral,yan2016robustness}.

\begin{theorem}
The number of misclassification nodes is bounded by $ 64\mmax \|\xh-X_0\|_F^2$.
\end{theorem}

\begin{proof}
Let $\hat{U}$ be the top $\cl$ eigenvectors of $\xhat$, $U\in \bR^{n\times \cl}$ be the top $\cl$ eigenvector of $X_0$. Let $\nu\in \bR^{\cl\times \cl}$ be the population value of the eigenvector corresponding to each cluster, $U=Z \nu$. 
By Davis-Kahan theorem~\cite{yu2014useful}, we have
\ba{
\|\hat{U}-U O\|_F^2 \le \frac{8\|\xh-X_0\|_F^2}{(\theta_{\cl}(X_0)-\theta_{\cl+1}(X_0))^2} = 8\|\xh-X_0\|_F^2
\label{eq:davis}
}
Define $\cM = \{i: \|c_i-Z_i\nu O\| \ge \frac{1}{\sqrt{2\mmax}}\}$. We now prove that $\cM$ is a superset of all misclassified nodes by the above procedure, and its cardinality is bounded as in the theorem statement.
$U$ is a unit basis so we know $I=U^TU=\nu^T Z^T Z \nu=\nu^T\diag(m_1,\cdots, m_r)\nu$. So $\theta_{\min}(\nu^T\nu) \ge \frac{1}{\mmax}$.

Define $\mathcal{C}=\{ M\in \bR^{n\times \cl}: M \text{ has no more than }\cl \text{ unique rows} \}$. Then minimizing the \km objective for $\hat{U}$ is equivalent to 
\[ \min_{\{s_1,\cdots,s_\cl\}\subset \bR^\cl} \sum_i \min_g \|\hat{u}_i - s_g\|_2^2 =\min_{M\in \mathcal{C}} \|\hat{U}-M\|_F^2 \]
So $C=[c_1,\cdots, c_{n}]=\arg\min_{M\in \mathcal{C}} \|\hat{U}-M\|_F^2$ and $\|C-\hat{U}\|\le\|Z\nu O-\hat{U}\|$. $c_i$ is the center assigned to point $i$ by running \km on $\hat{U}$.

Now we prove all points lying outside of $M$ is correctly labeled, or equivalently, $\|c_i-Z_i\nu O\|<\|c_i-Z_j\nu O\|_2$ for all $Z_j\ne Z_i$. To see this, note for $\forall i,j\in [n]$, when $Z_i\ne Z_j$, 
\begin{align*}
\|Z_i\nu-Z_j\nu\| =& \|(Z_i-Z_j)\nu\| \ge \sqrt{2}\min_{x:\|x\|^2=1} \sqrt{x^T\nu^T\nu x} \ge \sqrt{\frac{2}{\mmax}}
\end{align*}
So 
\begin{align}
\|c_i-Z_j\nu O\|_2 \ge \|Z_i\nu-Z_j\nu\| -\|c_i-Z_i\nu O\|\ge \sqrt{\frac{2}{\mmax}}-\sqrt{\frac{1}{2\mmax}}=\sqrt{\frac{1}{2\mmax}}
\end{align}
Therefore when $Z_i\ne Z_j$, $\|c_i-Z_i\nu O\|<\sqrt{\frac{\cl}{2n}}\Rightarrow \|c_i-Z_i\nu O\|_2<\|c_i-Z_j\nu O\|_2$, which means node $i$ is correctly clustered.

Below we bound the cardinality of $\cM$. By Markov's inequality,
\bas{
|\cM| \le& 2\mmax \sum_{i\in [n]} \| c_i-Z_i\nu O \|_F^2 \\
=& 2\mmax \|C -U O\|_F^2\\
 \le& 2\mmax (\|C -\hat{U} \|_F+\|\hat{U}-U O\|_F)^2
 }
 Note 
 \bas{
\|C - \hat{U} \|_F^2 \le \|\hat{U}-UO\|_F^2
}
Therefore, we have
\ba{
 |\cM| \le 8\mmax \|\hat{U}-U O\|_F^2\label{eq:u_bound}
 }

Combining with Eq.~\eqref{eq:davis}, we have
\bas{
|\cM| \le 64\mmax \|\xh-X_0\|_F^2
}
\end{proof}

\end{document}